\documentclass[unicode,orivec]{llncs}

\usepackage[T1]{fontenc}
\usepackage{graphicx}
\usepackage[linesnumbered,algoruled,vlined]{algorithm2e}
\DontPrintSemicolon
\SetArgSty{}
\SetAlCapNameSty{textsf}
\SetKw{KwOr}{or}
\SetKw{KwAnd}{and}
\SetKw{KwNot}{not}
\SetKwProg{Procedure}{Procedure}{}{}
\usepackage{amsmath}

\usepackage{amsthm}
\usepackage[unicode=true,psdextra]{hyperref}
\usepackage{tcolorbox}
\usepackage{amssymb}
\usepackage{booktabs}
\usepackage{array}
\usepackage{enumitem}
\usepackage{cite}
\usepackage[full,disableredefinitions,small]{complexity}
\usepackage{xspace}
\usepackage{optidef} %
\usepackage{xpatch}
\usepackage{xstring}
\ExpandArgs{c}\xpatchcmd{environment mini code}{minimize}{minimise}{}{}
\ExpandArgs{c}\xpatchcmd{environment mini code}{minimize}{minimise}{}{}
\ExpandArgs{c}\xpatchcmd{environment mini code}{minimize}{minimise}{}{}
\ExpandArgs{c}\xpatchcmd{environment mini! code}{minimize}{minimise}{}{}
\ExpandArgs{c}\xpatchcmd{environment mini! code}{minimize}{minimise}{}{}
\ExpandArgs{c}\xpatchcmd{environment maxi code}{maximize}{maximise}{}{}
\ExpandArgs{c}\xpatchcmd{environment maxi* code}{maximize}{maximise}{}{}
\ExpandArgs{c}\xpatchcmd{environment maxi* code}{maximize}{maximise}{}{}
\ExpandArgs{c}\xpatchcmd{environment maxi code}{maximize}{maximise}{}{}
\usepackage[capitalize]{cleveref}
\newtheorem{observation}{Observation}
\theoremstyle{plain}
\usepackage{orcidlink}
\usepackage{bbding}

\SetCommentSty{mycommfont}

\definecolor{dark blue}{rgb}{0.121,0.47,0.705}
\let\emph\relax\DeclareTextFontCommand{\emph}{\color{dark blue}\em}

\title{Moving Geometric Objects to Render Their Intersection Graph Connected or Locally Dense}

\titlerunning{Moving Geometric Objects}

\author{Tesshu~Hanaka\inst{1}\orcidlink{0000-0001-6943-856X} \and
Nicol\'as~{Honorato-Droguett}\inst{2}\,\orcidlink{0009-0005-1969-3649}\,\Envelope \and
Hirotaka~Ono\inst{2}\orcidlink{0000-0003-0845-3947} \and
Samuel~Wolf\inst{3}\orcidlink{0009-0009-7098-6147} \and
Alexander~Wolff\inst{3}\orcidlink{0000-0001-5872-718X}}

\authorrunning{T.~Hanaka, N.~{Honorato-Droguett}, H.~Ono, S.~Wolf,A.~Wolff}

\institute{Kyushu University, Fukuoka, Japan,
\email{hanaka@inf.kyushu-u.ac.jp} \and
Nagoya University, Nagoya, Japan,
\email{honorato.droguett.nicolas.n7@s.mail.nagoya-u.ac.jp},
\email{ono@i.nagoya-u.ac.jp} \and
Universit\"{a}t W\"{u}rzburg, W\"{u}rzburg, Germany,
\email{samuel.wolf@uni-wuerzburg.de}}

\newcommand{\gged}[1][]{%
  \textup{\textsc{GGED}\if\relax\detokenize{#1}\relax\else($#1$)\fi}\xspace
}
\newcommand{\ggedlong}{\textsc{Geometric Graph Edit Distance}\xspace}
\newcommand{\threepartition}{{\scshape 3-Partition}\xspace}
\newcommand{\partition}{{\scshape Partition}\xspace}
\newcommand{\planarmonothreesat}{{\scshape Planar Monotone 3-SAT}\xspace}

\newlength{\problabelwidth}
\newcommand{\defproblem}[4]{%
  \begin{tcolorbox}[left=1.4mm, right=0.9mm, top=0.5mm, bottom=0.5mm]
    \settowidth{\problabelwidth}{\bfseries Problem:~~}%
    \begin{tabular}{@{}>{\bfseries}l@{~~}p{\dimexpr\linewidth-\problabelwidth\relax}@{}}
      Problem: & #1\\[.1ex]
      Input: & #2\\[.1ex]
      #4: & #3
    \end{tabular}%
  \end{tcolorbox}%
}

\newcommand{\defdecproblem}[3]{\defproblem{#1}{#2}{#3}{Task}}
\newcommand{\defoptproblem}[3]{\defproblem{#1}{#2}{#3}{Output}}

\newcommand{\edgeless}{\ensuremath{\Pi_{\mathrm{edgeless}}}\xspace}

\newcommand{\kclique}[1][k]{\ensuremath{\Pi_{#1\mathrm{\text{-}clique}}}\xspace}
\newcommand{\kconnected}[1][k]{\ensuremath{\IfEqCase{#1}{%
        {1}{\Pi_{\mathrm{conn}}}}[\Pi_{#1\mathrm{\text{-}conn}}]}\xspace}

\providecommand{\A}{}\renewcommand{\A}{\ensuremath{\mathcal{A}}\xspace}

\renewcommand{\C}{\mathcal{C}}
\providecommand{\D}{}\renewcommand{\D}{\ensuremath{\mathcal{D}}\xspace}

\providecommand{\I}{}\renewcommand{\I}{\ensuremath{\mathcal{I}}\xspace}

\renewcommand{\S}{\ensuremath{\mathcal{S}}\xspace}

\newcommand{\X}{\ensuremath{\mathcal{X}}\xspace}

\newcommand{\onenorm}[1]{\lVert #1 \rVert}
\newcommand{\abs}[1]{\lvert #1 \rvert}
\newcommand{\len}[1]{\mathrm{len}(#1)}
\newcommand{\size}[1]{|#1|}
\newcommand{\set}[1]{\{#1\}}

\newcommand{\polytime}{\operatorname{poly}}
\makeatletter
\DeclareMathOperator*{\argmin}{\smash[b]{\operator@font arg\,min}}
\DeclareMathOperator*{\argmax}{\smash[b]{\operator@font arg\,max}}
\makeatother

\crefname{equation}{eq.}{eqs.}
\Crefname{equation}{Eq.}{Eqs.}
\crefname{algocf}{algorithm}{algorithms}
\Crefname{algocf}{Algorithm}{Algorithms}
\crefname{observation}{observation}{observations}
\Crefname{observation}{Observation}{Observations}

\spnewtheorem{claim}{Claim}{\bfseries}{\rmfamily}
\crefname{claim}{claim}{claims}
\Crefname{claim}{Claim}{Claims}

\usepackage{thm-restate}
\usepackage{apptools}

\newcommand\lncs[1]{}
\newcommand\arxiv[1]{#1}

\arxiv{\newcommand{\restateref}[1]{\IfAppendix{\hyperref[#1]{$\star$}}{\hyperref[#1*]{$\star$}}}}
\lncs{\newcommand{\restateref}[1]{\href{https://arxiv.org/abs/2609.22958}{$\star$}}}

\begin{document}

\maketitle

\begin{abstract}
    In this paper, we study graph editing problems on geometric
    intersection graphs.  For a tuple~$\S=(S_1,\dots,S_n)$ of geometric objects in some
    Euclidean space, let $G_\S$ be their intersection graph.
    We study the problem of
    finding a tuple~$D=(d_1,\dots,d_n)$ of movement vectors such
    that the resulting intersection graph~$G_{\S+D}$ (after moving,
    for every $i \in \{1,\dots,n\}$, object~$S_i$ by $d_i$)
    has a predefined property and the total moving
    distance $\|D\|$ is minimum.  In the weighted version, we are also given a weight vector $w=(w_1,\dots,w_n)$ with
    positive entries, and the objective is to minimise the total
    weighted moving distance $\|w \cdot D\|$.

    \quad We first consider the property {\em locally dense}, which we define
    as containment of a $k$-clique.
    Given $n$ weighted intervals, we solve the problem with respect to this property
    in $O(k^{1/3} n \log^{1+\varepsilon} n)$ time for any $\varepsilon>0$.
    We then consider {\em $k$-connectivity} for $1\le k \le n-1$. Given $n$ unweighted unit intervals, we solve the problem in $O(n^2 \log n)$ time and, for $k=1$, in $O(n\log n)$ time.
    For $k=1$, we prove strong \NP-hardness on intervals of arbitrary length and on weighted unit disks (with only two distinct weights), and weak \NP-hardness on weighted intervals (even when lengths equal weights).
\end{abstract}

\section{Introduction}

A \emph{(geometric) intersection graph} is a graph whose vertices correspond to some geometric objects and whose edges represent the intersection of pairs of objects.
Fundamental examples are interval graphs on the real line and disk graphs in the Euclidean plane.
These classes have been extensively studied because their geometric structure strongly affects the complexity of many well-known problems.
For example, \textsc{Maximum Clique}, \textsc{Graph Colouring}, and \textsc{Minimum Vertex Cover} are \NP-hard on general graphs~\cite{Garey1979} but polynomial-time solvable on interval graphs~\cite{Golumbic1980,Gupta1982}.
In contrast, many standard graph problems remain \NP-hard on disk graphs~\cite{Breu1998,Clark1990}~-- even on unit disk graphs.
It is thus natural to ask which problems remain tractable on specific intersection graph classes.

In recent work, graph modification has been studied on geometric intersection graphs under natural geometric edit operations.
Fomin, Golovach, Inamdar, Saurabh, and Zehavi~\cite{Fomin2023}
studied moving at most $k$ disks in a given set of unit disks by at most~$d$ so that the intersection graph becomes edgeless, eliminating all pairwise overlaps.
They obtained kernels and {\FPT}-algorithms parameterised by $k+d$.
The min-sum variant is called {\ggedlong} (\gged)~\cite{HonoratoDroguett2024,HonoratoDroguett2025,HonoratoDroguett2026}.
Given a collection of geometric objects of a specific class, how do we move the objects, minimising the total moving distance, such that their intersection graph becomes, say, edgeless, complete, or $k$-clique-free?
Formally, {\gged} is defined as follows (where, for a geometric object $S$ and a vector $d$, $S+d=\{s+d : s \in S\}$).

\defoptproblem{\ggedlong w.r.t.\ graph class~$\Pi$}%
{An $n$-tuple~$\S=(S_1,\dots,S_n)$ of geometric objects in
    $\mathbb{R}^d$, weights $w \in \mathbb{R}^n_{> 0}$.}%
{An $n$-tuple $D=(d_1,\dots,d_n) \in \mathbb{R}^{d \cdot n}$ such that
    $G_{\S+D} \in \Pi$ and $\|w \cdot D\|=\sum_{i=1}^n w_i \cdot \|d_i\|$
    is minimum.} %

We use \gged[\Pi] as shorthand for the above problem and we continue the above line of work for \emph{connectivity} and \emph{local density}.
Specifically, we move objects until their intersection graph is $k$-connected ($\Pi=\kconnected$) or contains a $k$-clique ($\Pi=\kclique$).
These problems arise in many domains.
In connectivity augmentation, the task is to add edges to a given graph to reach a connectivity threshold.
It is a classical problem in algorithmic graph theory~\cite{Eswaran1976} and has also been studied under (beyond-) planarity constraints,
both for geometric \cite{rw-acpgg-12,Toth-EJC12} and topological graphs
\cite{KantB-WADS91,Akitaya2025}.
\gged[\kconnected] is a geometric variant of connectivity augmentation.
Related movement formulations include pebble movement on graphs~\cite{Demaine2009,Demaine2014}, movement in geometric settings~\cite{Anari2016} and a min-max variant on a closed cycle~\cite{Li2026}.
Applications include sensor networks~\cite{Gupta1999} and wireless connectivity of mobile agents~\cite{Bredin2005,Wang2006,Zavlanos2007,Czyzowicz2010,AndrewsWang-Algorithmica17}.

Local density has also been well studied.
Given a set of $n$ weighted points, the \emph{capacitated geometric
    median problem}~\cite{Shenmaier2021,Shenmaier2021a} asks to
choose~$k$ points from the set and a centre~$c$ such that the total
weighted distance to $c$ is minimised.
A related problem is the \emph{smallest $k$-enclosing circle
    problem}~\cite{Efrat1994} where instead the maximum distance of the
selected points from the circle centre is minimised.

The above properties also arise specifically in one-dimensional models.
Sensing ranges on highways and DNA fragments in genomes are naturally modelled by intervals, and connectivity of their intersection graphs models communication, coverage, or chains of overlapping fragments~\cite{Czyzowicz2010,AndrewsWang-Algorithmica17,Yan2012,Zhang1994,Wu2014}.
In these models, \mbox{($k$-)connectivity} represents chains that remain connected after removing fewer than $k$ objects, while $k$-cliques represent $k$ objects with a common intersection.
Consequently, it is natural to study {\gged} even restricted to intervals.

\paragraph{Our Contribution.}
For \gged[{\kclique}], we give an $O(k^{1/3}n\log^{1+\varepsilon}n)$-time algorithm on $n$ weighted intervals of arbitrary length, for any $\varepsilon>0$ (\cref{sec:k_clique}), and an $O(n^2\log n)$-time algorithm for \gged[{\kconnected[k]}] on $n$ unweighted unit intervals. We improve the latter to $O(n\log n)$ time for $k=1$; see \cref{sec:unweighted_unit_intervals}.

We also prove several hardness results for \gged[{\kconnected[1]}] (\Cref{sec:hardness}), where $\kconnected[1]$ is a shorthand for~$\Pi_{\text{1-conn}}$.
We show that the problem is strongly \NP-hard on unweighted intervals of arbitrary length, weakly \NP-hard on weighted intervals even when length equals weight, and strongly \NP-hard on weighted unit disks with only two weights.
We also prove weak \NP-hardness for the {\em edgeless} case and para-\NP-hardness with respect to the number of maximal cliques, even when there is only one, resolving two
open problems of~\cite{HonoratoDroguett2026}.
\Cref{tab:summary} summarises the new and existing results.
We start with preliminaries (\cref{sec:preliminaries}) and close with open problems (\cref{sec:open}).
Proofs of statements marked by~($\star$) appear in the \arxiv{appendix}\lncs{full version of this paper~\cite{Hanaka2027}}.

\begin{table}[bt]
    \centering
    \caption{Results for \kconnected and \kclique. The bound in row~2
        holds for every $\varepsilon>0$. `\#Weights' is the number of
        distinct weights (`1' denotes unit weights).}
    \label{tab:summary}
    \begin{tabular}{@{}l@{\quad}l@{~~}c@{\quad}c@{\quad}c@{}}
        \toprule
        \bf Graph Class & \bf Object Type                     & \bf \#Weights & \bf Complexity                     & \bf Reference
        \\\midrule
        \kclique        & Unit Interval                       & 1             & $O(n\log n)$
                        & \cite[Thm.~4]{HonoratoDroguett2024}                                                                                                   \\
                        & Interval                            & $n$           & $O(k^{1/3}n\log^{1+\varepsilon}n)$
                        & Theorem\hfill\ref{thm:k_clique}                                                                                                       \\
        \kclique[n]     & Interval                            & $n$           & $O(n)$
                        & \cite[Thm.~1]{HonoratoDroguett2024}                                                                                                   \\
        \midrule
        \kconnected     & Unit Interval                       & 1             & LP ($\polytime$)
                        & \cite[Thm.~8]{HonoratoDroguett2024}                                                                                                   \\
                        & Unit Interval                       & 1             & $O(n^2\log n)$
                        & Theorem\hfill\ref{thm:kconn_nlogn}                                                                                                    \\
        \kconnected[1]  & Unit Interval                       & 1             & $O(n\log n)$
                        & Theorem\hfill\ref{thm:kconn_nlogn}                                                                                                    \\
                        & Interval                            & 1             & strongly \NP-hard                  & Theorem\hfill\ref{thm:snph_intv}           \\
                        & Interval ($\len{I_i}=w_i$)          & $n$           & weakly \NP-hard                    & Theorem\hfill\ref{thm:wequall_conn_nphard} \\
                        & Unit Disk                           & 2             & strongly \NP-hard                  & Theorem\hfill\ref{thm:snph_wdisk}          \\
        \bottomrule
    \end{tabular}
\end{table}

\section{Preliminaries}
\label{sec:preliminaries}

For $n\in \mathbb{Z}^+$, we use $[n]$ and $[n]_0$ as shorthand for the
sets $\set{1,\ldots,n}$ and $\set{0,\ldots,n}$, respectively.  The
\emph{distance} between two points $p$ and $q$ in~$\mathbb{R}$ is
$\abs{p-q}$.  The \emph{$L_2$ distance} (Euclidean distance) between two
points $(p_x,p_y)$ and $(q_x,q_y)$ in~$\mathbb{R}^2$ is
$\lVert p-q\rVert_2 = \sqrt{(p_x-q_x)^2+(p_y-q_y)^2}$.

\paragraph{Geometry.}
Given two points $a,b\in \mathbb{R}$ with $a\le b$, the \emph{interval} defined by $a$ and $b$ is the set $I = [a,b] = \set{x\in \mathbb{R}\colon a\le x\le b}$.
We call $a$ and $b$ the \emph{endpoints} of $I$; $\len{I} = |b-a|$ is the \emph{length} of $I$.
If $\len{I} = 1$, then $I$ is a \emph{unit interval}.
The point $(a+b)/2$ is the \emph{centre} of $I$ and is denoted by $c(I)$.
Given a number $r>0$ and a point $p\in \mathbb{R}^2$, the set
$\set{x\in \mathbb{R}^2\colon \lVert x-p \rVert_2 \le r}$ is the
radius-$r$ \emph{disk} centred at~$p$.  Since it slightly simplifies
the description of our results, we call disks with {\em diameter}~1
(that is, radius-1/2 disks) \emph{unit disks}.

\paragraph{Graphs.}
We consider only simple, finite, and undirected graphs.  For a
graph~$G$, let $V(G)$ be the vertex set and let
$E(G) \subseteq {V(G) \choose 2}$ be the edge set of~$G$.
Given a tuple $\S = (S_1,\ldots,S_n)$ of geometric objects in $\mathbb{R}^d$, the \emph{geometric intersection graph} of~\S, $G_\S$, has a vertex for each element of~\S and an edge between two elements if they intersect.
If \S is a tuple of (unit) intervals, then $G_\S$ is a \emph{(unit) interval graph}.
Similarly, if~\S is a tuple of (unit) disks, then $G_\S$ is a \emph{(unit) disk graph}.
In this paper, we focus on the graph classes $\kclique = \{G : G\text{ has a $k$-clique}\}$ and $\kconnected = \{G : G\text{ is $k$-connected}\}$.

We assume the following known properties of instances of
intervals~\cite{HonoratoDroguett2025,HonoratoDroguett2026} throughout
the paper.  Due to our objective function, they hold for both
\gged[{\kconnected}] and \gged[\kclique].
\begin{itemize}
    \item There exists an optimal solution where at least one movement
          vector is zero.
    \item In the case of intervals with unit length and unit weight, there
          exists an optimal solution that preserves the initial order of the
          objects.
\end{itemize}

\section{\texorpdfstring{\gged[\kclique]}{GGED(Pi-kclique)} on Weighted Intervals of Arbitrary Length}
\label{sec:k_clique}

In this section, we present a subquadratic-time algorithm for
\gged[{\kclique}] on weighted intervals, for any $1 < k < n$.  For
$k = n$, the problem can be solved in $O(n)$
time~\cite[Thm.~1]{HonoratoDroguett2024}, and for $k=1$, for every
non-empty instance~$\I$, the graph $G_\I$ already contains a 1-clique.
Since there is an optimal solution that is attained at an endpoint, a
straightforward quadratic-time algorithm evaluates each endpoint in
linear time: it uses the $k$-selection algorithm~\cite{Blum1973} and
then sums up the $k$ smallest weighted moving distances.
The key insight for a faster algorithm is a shift of perspective.
Instead of considering intervals directly,
we translate our problem into a two-dimensional problem on line segments.
To this end, we consider functions $f_i \colon \mathbb{R} \to \mathbb{R}_{\ge0}$ for each interval
$I_i \in \I$ with weight $w_i$, where $f_i(x)$ represents the cost of moving $I_i=[a_i, b_i]$ to contain $x$.
More precisely, let
\begin{align*}
    f_i(x) = \begin{cases}
                 w_i(a_i-x) , & \text{ if } x < a_i             \\
                 0,           & \text{ if } a_i \leq x \leq b_i \\
                 w_i(x-b_i),  & \text{ otherwise.}
             \end{cases}
\end{align*}
Note that $f_i$ is a piecewise linear convex and continuous function that can be interpreted as an unbounded x-monotone convex chain (of line segments).
Two functions $f_i$ and $f_j$ cross each other (change sides) at most twice;
that is, they are \emph{pseudo-parabolas}.
To solve \gged[\kclique], we want to know, for every $x\in \mathbb{R}$, the set of indices $S_x \subseteq [n]$ that satisfies: $i \in S_x$ if and only if $f_i(x)$ is among the $k$ smallest values in the set $\{ f_i(x) \mid I_i \in \I\}$.

For a tuple $\I = (I_1,\dots,I_n)$ of intervals,
the functions $f_1,\ldots,f_n$ form an arrangement of x-monotone curves.
In such an arrangement, the \emph{$k$-level} is an x-monotone sequence of curve segments such that, for every $x$, exactly $k$ curves lie on or below it; see \cref{fig:klevel_intv} for an example.
Then in the arrangement of our pseudo-parabolas, the $k$-level gives us the set $S_x$ for every $x \in \mathbb{R}$.
Thus the minimum of $\sum_{i \in S_x} f_i(x)$ over all $x$ equals the minimum cost of moving intervals in~$\I$ so that the resulting intersection graph contains a $k$-clique.

\begin{figure}[bt]
    \centering
    \includegraphics[scale=1,page=2]{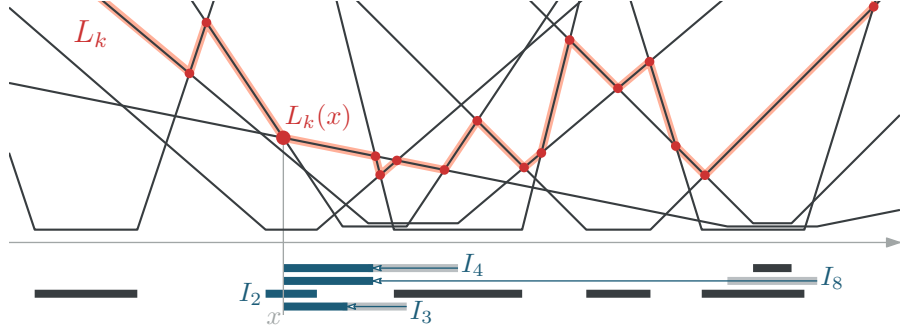}%
    \caption{The $k$-level $L_k$ (bold red x-monotone curve) for an arbitrary tuple of nine intervals for $k =4$.
        The four cheapest intervals to gather at point $x$ are $I_2,I_3,I_4,I_8$.
        The functions are shifted by the symbolic value $\delta_i$.
        There are $16$ vertices of $L_k$ shown.}
    \label{fig:klevel_intv}
\end{figure}

Let $H$ be the arrangement of the pseudo-parabolas corresponding to the functions $f_1,\dots,f_n$ and $L_k$ its $k$-level.
We write $L_k(x)$ for its value at x-coordinate~$x$ and $|L_k|$ for its \emph{complexity} (i.e., the number of vertices of $L_k$).

Using the $k$-level, we can show the following result.

\begin{theorem}\label{thm:k_clique}
    Given a tuple of $n$ weighted intervals, %
    \gged[\kclique] can be solved in $O(nk^{1/3}\log^{1+\varepsilon}n)$ time for an arbitrarily small $\varepsilon >0$.
\end{theorem}
\begin{proof}
    Suppose we have already computed the $k$-level of the functions (segments)~$f_i$. Then,
    we claim the following.

    \begin{restatable}[\restateref{claim:k_clique_given_k_level}]{claim}{kcliquegivenklevel}
        \label{claim:k_clique_given_k_level}
        Given the $k$-level $L_k$ and endpoints of $\I$ sorted, \gged[\kclique]
        can be solved in $O(n+|L_k|)$ time.
    \end{restatable}
    \begin{proof}[Proof Sketch]
        \renewcommand{\qedsymbol}{$\triangle$}
        We describe an algorithm that computes an optimal point~$x^\star$ and a set~$S^\star$ of $k$ (interval) indices minimising $\sum_{i\in S^\star}f_i(x^\star)$.
        The algorithm sweeps the \emph{event points} (the x-coordinates of vertices of~$L_k$ and interval endpoints) from left to right.
        There are at most $|L_k|+2n$ of these points, and we obtain them in $O(n+|L_k|)$ time by merging the vertices of $L_k$ and the sorted endpoints of intervals.
        It maintains
        a candidate set~$S$ of interval indices, and its objective function $\sum_{i \in S} f_i$
        using two variables~$m$ and~$t$ that store the sum of slopes and the sum of y-intercepts of the functions~$f_i$ for $i\in S$, respectively.
        Note that between consecutive event points, $S$ is fixed and
        $\sum_{i\in S}f_i$ is linear, so it suffices to evaluate the objective at event points.
        Further, note that the set~$S$ changes only at the x-coordinates of vertices of~$L_k$, and $x^\star$ lies between the leftmost and rightmost endpoints of~$\I$.
        Additionally, $m$ and $t$ may only change if a new function is added to~$S$ or an endpoint of an interval is reached.
        Consequently, it suffices to process the event points in the range spanned by the leftmost and rightmost endpoints of $\I$.
        At each event point $x$, we update $S$, $m$, and~$t$ in $O(1)$ time. We evaluate $mx+t$ and update $x^\star$ whenever this improves the best value found so far. Hence the sweep takes $O(n+|L_k|)$ time.
        After the sweep, we recover $S^\star$ in $O(n)$ time by selecting the indices of the $k$ smallest values $f_i(x^\star)$.
    \end{proof}

    It remains to compute the $k$-level of~$H$ using an appropriate representation.
    Our arrangement~$H$ consists of chains, each represented by two rays and a horizontal line segment.
    First, we can assume that no horizontal segment is part of the $k$-level of $H$ since otherwise $G_\I$ already contains a $k$-clique.
    This can be checked in $O(n\log n)$ time (it is also possible to introduce a symbolic shift $\delta_i$ for each $f_i$ to ensure that two functions intersect at most twice based on a linear order derived from the containment relationship of the horizontal segments).

    We convert these functions into line segments by clipping the rays with an axis-aligned bounding box.
    The box spans an x-range from the leftmost endpoint~$a$ of~$\I$ to the rightmost endpoint~$b$ of~$\I$ and a y-range from $0$ to $(\max_{i \in [n]}\{w_i\})(b-a)$.
    By construction, this box contains all crossings within the x-range and the optimal point must lie in the range $[a, b]$.
    We can now employ an algorithm that computes the $k$-level~$L_k$ of line segments.

    In his unpublished manuscript, Chan \cite{chan99-klevel} remarks that his algorithm for computing the $k$-level of line arrangements runs deterministically in $O((n + |L_k|)\log^{1+\varepsilon}n)$ time (for some arbitrarily small $\varepsilon >0$) and can be adapted to compute the $k$-level of line segments.\footnote{Alternatively, {Har-Peled's} algorithm~\cite{HarPeled2000} computes the $k$-level of arcs that pairwise intersect at most $t$ times, in expected $O(\lambda_{t+2}(n + |L_k|)\cdot\log n)$ time, where $\lambda_{t+2}(n + |L_k|)$ is the length of a {\em Davenport--Schinzel sequence} of order~$t+2$ having~$n + |L_k|$ symbols. For us $t=2$, and it is known that $\lambda_{4}(n + |L_k|)\in O((n+|L_k|)\cdot 2^{\alpha(n+|L_k|)})$, where $\alpha$ is the inverse Ackermann function~\cite{davenport-schinzel-bounds}; this gives the same bound in expected time. An earlier algorithm by Edelsbrunner and Welzl~\cite{ew-cb2daa-86} runs in $O(|L_k|\log^2n)$ time and can be adapted to arrangements of x-monotone convex chains~\cite{EverettRvK-IJCGA96}.}

    Using this algorithm, together with the upper bound for the complexity of the $k$-level of line segments by Dey~\cite{Dey1998,Agarwal1998a}, and~\Cref{claim:k_clique_given_k_level}, we obtain a runtime of $O(nk^{1/3}\log^{1+\varepsilon}n)$ to solve \gged[\kclique] on weighted intervals. %
\end{proof}

\section{\texorpdfstring{\gged[{\kconnected}]}{GGED(Pi-k-conn)} on Unweighted Unit Intervals}
\label{sec:unweighted_unit_intervals}

In this section, we study \gged[{\kconnected}] on unweighted unit intervals for $k\in [n-1]$ by relating our problem to variants of \textsc{Isotonic Regression}.
This problem asks, given a tuple of numbers $(a_1, \dots, a_n)$, to find values $z_1,\dots,z_n$
subject to $z_i \leq z_{i+1}$ minimising the $L_1$-error $\sum_{i=1}^n |a_i-z_i|$.
We use the following characterisation of unit-interval graphs in~\kconnected.
\begin{lemma}[{\!\!\cite[Lem.~10]{HonoratoDroguett2024}}]
    \label{lem:kconnectediff}
    Given a (sorted) tuple
    $\I=(I_1,\dots,I_n)$ of unit intervals, $G_\I \in \kconnected$ if and only if $I_i\cap I_{i+k}\neq\emptyset$ for every $i\in[n-k]$.
\end{lemma}

Using this lemma, the authors of~\cite{HonoratoDroguett2024} gave the following polynomial-time program for \gged[{\kconnected}] on unweighted unit intervals:
Minimise $\sum_{i=1}^n\abs{d_i}$ subject to $(c(I_{i+1})+d_{i+1})-(c(I_i)+d_i)\ge0$ for $i\in[n-1]$ and $(c(I_{i+k})+d_{i+k})-(c(I_i)+d_i)\le1$ for $i\in[n-k]$, where $d_i$ is the movement applied to~$I_i$.
While this result provides a polynomial-time algorithm, the runtime depends on a polynomial
of high degree, and linear programs may face numerical issues.
We instead give faster combinatorial algorithms for \gged[{\kconnected}].

\begin{theorem}\label{thm:kconn_nlogn}
    On unweighted unit intervals, \gged[{\kconnected}] can be solved in
    $O(n^2\log n)$ time if $k>1$ and in $O(n\log n)$ time if $k=1$.
\end{theorem}
\begin{proof}
For $k=1$, it turns out that the formulation of~\cite{HonoratoDroguett2024} reduces \gged[{\kconnected[1]}] to \textsc{Isotonic Regression} as follows.
Consider $d_1,\ldots,d_n$ and set $z_i = -(c(I_i) +d_i -i)$.
Then the constraint $(c(I_{i+1}) + d_{i+1}) - (c(I_i) +d_{i}) \le 1$ becomes $z_{i} \le z_{i+1}$ since $c(I_i) + d_i = -(z_i-i)$.
Moreover, $d_i$ is given by $-c(I_i) - z_i +i$.
Thus, we set $a_i = -c(I_i) + i$ and solve \textsc{Isotonic Regression} for $a_1,\ldots,a_n$.
This takes $O(n\log n)$ time~\cite{irdp-r-sosa19}, which
dominates the runtime for \gged[{\kconnected[1]}].
The constraint $(c(I_{i+1}) + d_{i+1}) - (c(I_i) +d_{i}) \ge 0$ becomes $z_{i+1} - z_{i} \le 1$, which holds for every optimal solution since $z_i < z_{i+1}$ implies $a_i \le z_i$ and $z_{i+1} \le a_{i+1} \le a_i + 1$.

For $k\geq 2$, we again use \textsc{Isotonic Regression}, adapting the program of~\cite{HonoratoDroguett2024}, and obtain the following linear program to choose final centres $x_i = c(I_i) + d_i$.
Recall that the order-preserving property allows us to assume $x_1\le\cdots\le x_n$.
\begin{mini!}
{}{\sum_{i=1}^n y_i\tag{$\textsf{P}$}}{}{}\label{eq:lp_k_connected}
\addConstraint{-x_i+y_i}{\ge -c(I_i)}{\quad\forall i \in [n]}\tag{$\textsf{P}.1$}\label{eq:lp_k_connected_1}
\addConstraint{x_i+y_i}{\ge c(I_i)}{\quad\forall i \in [n]}\tag{$\textsf{P}.2$}\label{eq:lp_k_connected_2}
\addConstraint{x_{i+1}-x_i}{\ge 0}{\quad\forall i \in [n-1]}\tag{$\textsf{P}.3$}\label{eq:lp_k_connected_3}
\addConstraint{x_i-x_{i+k}}{\ge -1}{\quad\forall i \in [n-k]}\tag{$\textsf{P}.4$}\label{eq:lp_k_connected_4}
\addConstraint{x_i}{\in\mathbb{R}}{\quad\forall i \in [n]}\tag{$\textsf{P}.5$}\label{eq:lp_k_connected_5}
\end{mini!}

Note that the first two constraints and the objective force $y_i=\abs{x_i-c(I_i)}$.

We note that~\eqref{eq:lp_k_connected} is a special case of the general model of Jewell~\cite{Jewell1975}, who formulated \textsc{Isotonic Regression} as a linear program and represented its dual as a min-cost circulation (without giving an explicit algorithm).
The authors of~\cite{Angelov2006} adapted this approach to weighted \textsc{Isotonic Regression} under %
$m$ partial-order constraints and gave an algorithm running in $O(nm+n^2\log n)$ time.
In our case, we have constraint~\eqref{eq:lp_k_connected_4} in addition to the constraints of \textsc{Isotonic Regression},
which is not captured by the model of~\cite{Angelov2006} for $k \geq 2$.
In the following, we adapt the approach in~\cite{Angelov2006} to also capture~\eqref{eq:lp_k_connected_4} (for the unweighted case).
We derive the dual and its min-cost circulation explicitly, resulting in a network with $n+2$ vertices and $4n-k$ arcs.
This gives us the necessary tools to prove the $O(n^2\log n)$ running time and to recover the final centres.

To this end, let $\alpha,\beta \in \mathbb{R}^{n}_{\ge 0}$,
$\lambda \in \mathbb{R}^{n-1}_{\ge 0}$, and $\mu \in \mathbb{R}^{n-k}_{\ge 0}$ be the dual variables of the constraints
of~\eqref{eq:lp_k_connected_1},~\eqref{eq:lp_k_connected_2},~\eqref{eq:lp_k_connected_3} and~\eqref{eq:lp_k_connected_4}, respectively.
Dualising~\eqref{eq:lp_k_connected} and negating its objective yields:
\begin{mini!}
{}{\sum_{i=1}^n c(I_i)(\alpha_i-\beta_i)+\sum_{i=1}^{n-k}\mu_i\tag{$\textsf{D}$}}{}{}\label{eq:dual_k_connected}
\addConstraint{-\alpha_i+\beta_i-\lambda_i+\lambda_{i-1}-\mu_{i-k}+\mu_i}{=0}{\quad\forall i\in[n]}\tag{$\textsf{D}.1$}\label{eq:dual_k_connected_1}
\addConstraint{\alpha_i+\beta_i}{\le 1}{\quad\forall i\in[n]}\tag{$\textsf{D}.2$}\label{eq:dual_k_connected_2}
\addConstraint{\alpha,\beta,\lambda,\mu}{\ge 0}{}\tag{$\textsf{D}.3$}\label{eq:dual_k_connected_3}
\end{mini!}
where $\lambda_i=0$ for $i\notin[n-1]$ and $\mu_i=0$ for $i\notin[n-k]$.
Note that~\eqref{eq:dual_k_connected_1} is a flow conservation constraint interpreting
$\alpha_i$, $\lambda_i$, and $\mu_{i-k}$ as flow entering a vertex $v_i$ and $\beta_i$, $\lambda_{i-1}$, and $\mu_i$ as flow leaving $v_i$.
Also, the flow corresponding to $\lambda_i$ and $\lambda_{i-1}$ does not
appear in the objective function.
Since $(s, v_i)$ and $(v_i, t)$ have opposite costs and $(t, s)$ has zero cost, we can suitably shift an optimal flow satisfying~\eqref{eq:dual_k_connected_1} to satisfy~\eqref{eq:dual_k_connected_2}.

Using this interpretation, we construct a flow network $N_k$ with vertices $\{v_1,\dots,v_n\}\cup\{s,t\}$, writing $p(e)$ and $u(e)$ for the cost and capacity of an arc~$e$.
Since $\alpha_i$ and $\beta_i$ appear only in the $i$th constraint, they
correspond to arcs $(s, v_i)$ and $(v_i, t)$, respectively.
Similarly, we add arcs $(v_{i+1},v_i)$ and $(v_i,v_{i+k})$ corresponding to $\lambda_i$ and $\mu_i$, respectively.
Lastly, we add the arc $(t, s)$.
We call the arcs among $v_1,\ldots,v_n$ the \emph{constraint arcs}.
All arc costs are given by the objective of~\eqref{eq:dual_k_connected}.
Finally, arcs $(s, v_i)$ and $(v_i, t)$ receive a capacity of 1, while all other arcs remain unlimited.
The capacities of $(s, v_i)$ and $(v_i, t)$ correspond to the upper bound given by~\eqref{eq:dual_k_connected_2}, whereas $\lambda_i$ and $\mu_i$ remain unlimited because these    variables are unbounded.
\Cref{fig:nk_skeleton} illustrates~$N_k$.

\begin{figure}[bt]
    \centering
    \includegraphics[scale=1,page=3]{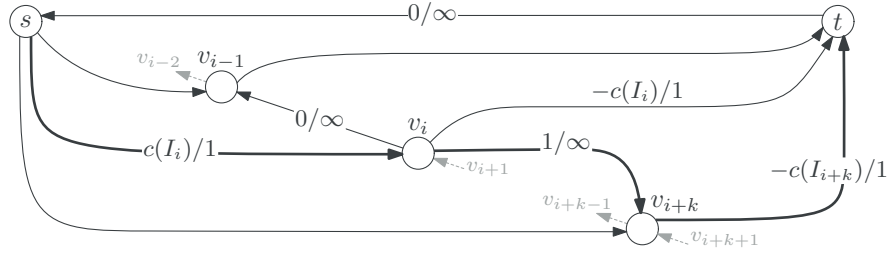}%
    \caption{ Network~$N_k$, with arc labels $p(e)/u(e)$.
        The bold path $s\to v_i\to v_{i+k}\to t$ has cost
        $c(I_i)+1-c(I_{i+k})$, which is negative exactly when
        $c(I_{i+k})-c(I_i)>1$.
    }
    \label{fig:nk_skeleton}
\end{figure}

\begin{restatable}[\restateref{claim:circ_equal_dual_opt}]{claim}{circequaldualopt}
    \label{claim:circ_equal_dual_opt}
    The cost of an adjusted min-cost circulation of $N_k$ is equal to the optimal value of~\eqref{eq:dual_k_connected}.
\end{restatable}

By \Cref{claim:circ_equal_dual_opt}, a min-cost circulation in~$N_k$ yields an optimal $(\alpha,\beta,\lambda,\mu)$.
Assume that we have such a solution.
Complementary slackness states that a positive dual (primal) variable forces the corresponding primal (dual) constraint to be tight.
In particular, $y_i = 0$ or $\alpha_i + \beta_i -1 = 0$.
First, if $\alpha_i+\beta_i=1$, the flow adjustment implies $(\alpha_i,\beta_i)\in\{(1,0),(0,1)\}$.
For $(\alpha_i,\beta_i) = (1,0)$, primal complementary slackness yields $y_i + c(I_i) = x_i$.
Now, $x_i \ge c(I_i)$ holds since $y_i\ge 0$.
For $(\alpha_i,\beta_i) = (0,1)$, primal complementary slackness yields $x_i+y_i = c(I_i)$ and hence $x_i \le c(I_i)$.
If $\alpha_i + \beta_i <1$, dual complementary slackness yields $y_i = 0$ and the first two primal constraints imply $x_i=c(I_i)$.
Combining these conditions with~\eqref{eq:lp_k_connected} yields the following \emph{system of difference constraints}:
\begin{equation*}
    \begin{array}{
            r@{\;}c@{\;}l@{\quad}l@{\qquad}
            r@{\;}c@{\;}l@{\quad}l
        }
        x_i-x_{i+1} & \le                           & 0,  & i\in[n-1]
                    & x_0-x_i                       & \le & -c(I_i)
                    & \text{if }\alpha_i=1,                                        \\
        x_{i+k}-x_i & \le                           & 1,  & i\in[n-k]
                    & x_i-x_0                       & \le & c(I_i)
                    & \text{if }\beta_i=1,                                         \\
        x_i-x_{i+1} & =                             & 0   & \text{if }\lambda_i>0,
                    & x_i-x_0                       & \le & c(I_i)
                    & \text{if }\alpha_i+\beta_i<1,                                \\
        x_{i+k}-x_i & =                             & 1   & \text{if }\mu_i>0,
                    & x_0-x_i                       & \le & -c(I_i)
                    & \text{if }\alpha_i+\beta_i<1,
    \end{array}
\end{equation*}
where $x_0$ is a dummy variable to transform (in)equalities of the form $x_i \le b$.
The above system has $O(n)$ constraints, and we can use the Bellman--Ford algorithm to solve the system in $O(n^2)$ time~\cite{cormen2009}.
Since the system is invariant under translation~\cite[Lem.~24.8]{cormen2009}, we translate a solution $x=(x_0,\ldots,x_n)$ by $-x_0$.
This makes $x_0=0$, so constraints $x_i - x_0 \le c(I_i)$ imply $x_i \le c(I_i)$ (and similarly for the other constraints).
Lastly, we set $y_i = \abs{x_i - c(I_i)}$ for every $i\in [n]$.

The obtained solution $(x,y)$ %
is feasible in~\eqref{eq:lp_k_connected} since $y_i = \abs{x_i - c(I_i)}$ ensures~\eqref{eq:lp_k_connected_1} and \eqref{eq:lp_k_connected_2} for all $i\in [n]$ and the system includes~\eqref{eq:lp_k_connected_3} and~\eqref{eq:lp_k_connected_4}.
Moreover, it satisfies complementary slackness with the optimal dual solution, hence strong duality implies its optimality for~\eqref{eq:lp_k_connected}.
It remains to compute a min-cost circulation in $N_k$. %

\begin{restatable}[\restateref{claim:min_cost_circulation}]{claim}{mincostcirculation}
    \label{claim:min_cost_circulation}
    A min-cost circulation in $N_k$ can be computed in $O(n^2\log n)$ time.
\end{restatable}
\begin{proof}[Proof Sketch]
    \renewcommand{\qedsymbol}{$\triangle$}
    If we replace $(t,s)$ by an arc $(s,t)$ with cost~$0$ and capacity~$n$, a min-cost circulation in~$N_k$ is equivalent to a min-cost $s$--$t$ flow of value~$n$.
    Using the ideas of Angelov, Harb, Kannan, and Wang~\cite{Angelov2006}, we apply Dijkstra's algorithm to compute this flow, once per augmentation.
    Since we have unit capacities, we have at most $n$ augmentations, yielding a runtime of $O(n^2\log n)$.
\end{proof}

Combining \Cref{claim:min_cost_circulation} with the $O(n^2)$-time solution of the difference constraints yields an $O(n^2\log n)$-time algorithm for \gged[{\kconnected}] when $k\ge2$.
\end{proof}

\section{Connectivity for Intervals and Weighted Unit Disks is NP-hard}
\label{sec:hardness}

In this section,
we show that \gged[{\kconnected[1]}] is \NP-hard via a reduction from {\threepartition}, adapting the reduction used to prove that
\gged[\edgeless] is strongly \NP-hard on unweighted intervals~\cite{HonoratoDroguett2026}, where $\edgeless$ is the class of graphs with no edges.
A similar reduction for interval coverage appears in~\cite{Czyzowicz2010}.
In the following, we sketch our reductions; full proofs are in
\arxiv{Appendix~\ref{sec_apx:hardness_proofs}}\lncs{the full version~\cite{Hanaka2027}}.

\begin{restatable}[\restateref{thm:snph_intv}]{theorem}{snphintv}
    \label{thm:snph_intv}
    \gged[{\kconnected[1]}] is strongly \NP-hard on unweighted intervals.
\end{restatable}
\begin{proof}[Proof Sketch]
    We reduce from {\threepartition} (see~\Cref{fig:redexample_intv1conn}).
    \begin{figure}[tb]
        \centering
        \includegraphics[scale=1,page=4]{gged_kclique_kconn.pdf}%
        \caption{\gged[{\kconnected[1]}] on intervals of arbitrary length: Sketch of the reduction.
        }
        \label{fig:redexample_intv1conn}
    \end{figure}
    Given a multiset $A=\set{a_1,\ldots,a_{3m}}$ of positive integers and a bound $L$ such that $\sum_{i=1}^{3m}a_i=mL$ and $L/4<a_i<L/2$, {\threepartition} asks whether $A$ can be partitioned into $m$ triples, each summing to $L$.
    We construct an instance containing an interval $I_i$ of length $a_i$ for each $a_i$, and long chains of short unweighted intervals \(\mathcal C_0,\ldots,\mathcal C_m\) (components) costly to move.
    We arrange the components such that $I_1,\ldots,I_{3m}$ must fill the gaps between them.
    We show this is possible with cost less than $T=3m(3mL+L)/4$ if and only if $A$ has the required partition.
\end{proof}

We consider the special case of weighted intervals in which lengths equal weights.
We show that \gged[{\kconnected[1]}] is weakly \NP-hard on these instances.
Afterwards, we adapt the reduction to prove the same result for \gged[\edgeless].
The authors of~\cite{HonoratoDroguett2026} stated this case as open.
For connectivity, we assume that intervals are closed, whereas for the edgeless case we assume that intervals are open.

In both cases, we reduce from {\partition}.
Given a multiset $A=\{a_1,\ldots,a_n\}$ of positive integers with $\sum_{i=1}^n a_i=2L$, {\partition} asks whether $A$ can be partitioned into two multisets, each summing to~$L$.
For every $i\in[n]$, we introduce an interval $I_i$ with $c(I_i)=0$ and $\len{I_i}=w_i=a_i$. We also introduce a {\em divider} interval $I_c$ with $c(I_c)=0$ and $w_c=\len{I_c}=4L+1$, and set $T=w_cL+L^2$.
The intervals are placed either to the left or to the right of~$I_c$, and the cost threshold forces the total length on each side to be~$L$.
For a solution, let $S_r$ be the total length of the intervals placed to the right of $I_c$.
For the left side, we have $S_\ell=2L-S_r$. The skeleton is illustrated in~\Cref{fig:partition_example}.

\begin{figure}[b]
    \centering
    \includegraphics[page=5]{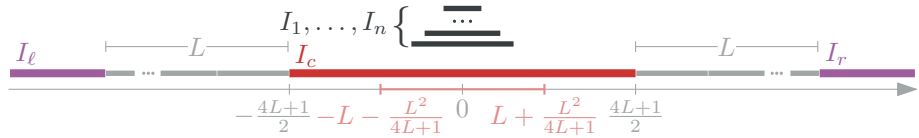}
    \caption{Skeleton of both reductions from {\partition}. The red interval is $I_c$, and its movement range under cost $T$ is highlighted in light red ($T/w_c= L+L^2/(4L+1)$). Each side of $I_c$ shows a subset of grey intervals with total length $L$. The intervals $I_\ell$ and $I_r$ are added only in the reduction for \gged[{\kconnected[1]}].}
    \label{fig:partition_example}
\end{figure}

\begin{restatable}[\restateref{lem:wequall_both_sides_cost}]{lemma}{wequallbothsidescost}
    \label{lem:wequall_both_sides_cost}
    Suppose the centre of $I_c$ is moved to point~$y$ and $I_c$ contains the origin.
    If the intervals on each side are placed consecutively with no gaps against~$I_c$, then the total weighted moving distance of $\I$ is at least $T+(S_r-L)^2$.
\end{restatable}

\begin{restatable}[\restateref{thm:wequall_conn_nphard}]{theorem}{wequallconnnphard}
    \label{thm:wequall_conn_nphard}
    \gged[{\kconnected[1]}] is weakly \NP-hard on tuples~$\I$ of intervals even if (i) $G_\I$ has three connected components, (ii) the length of each interval is equal to its weight and (iii) all lengths and weights are positive integers.
\end{restatable}
\begin{proof}[Proof Sketch]
    We add two intervals $I_\ell$ and $I_r$, each of length and weight $2T+1$, one on each side of $I_c$ with a gap of length $L$, as shown in \Cref{fig:partition_example}.
    We move the intervals corresponding to the elements of $A$ to close the two gaps.
    A valid partition yields a distance vector $D$ such that $G_{\I+D}\in\kconnected[1]$ and $\onenorm{w\cdot D}=T$.
    In the other direction, let $D$ be a distance vector such that $G_{\I+D}\in\kconnected[1]$ and $\onenorm{w\cdot D}\le T$.
    Any connected solution of cost at most $T$ must use $I_1,\ldots,I_n$ to form two chains between $I_\ell$, $I_c$, and $I_r$.
    We derive the cost of forming these chains with \Cref{lem:wequall_both_sides_cost}.
    If intervals are placed consecutively against $I_c$, the two chains can be shifted towards the origin, maintaining connectivity and decreasing the cost by at most $2L$ times the total moving distance of $I_\ell$ and $I_r$ (these intervals can certainly be moved to reduce the gaps).
    However, since $I_\ell$ and $I_r$ have weight $2T+1>2L$, their movement compensates for the possible decrease.
    Hence $\onenorm{w\cdot D}\ge T+(S_r-L)^2$
    and $S_r=L$ holds since $\onenorm{w \cdot D}\le T$.
    Therefore the intervals on either side correspond to a valid partition.
    Lastly, $G_\I$ has three connected components.
\end{proof}

\begin{restatable}[\restateref{thm:wequall_edgeless_nphard}]{theorem}{wequalledgelessnphard}
    \label{thm:wequall_edgeless_nphard}
    \gged[{\edgeless}] is weakly \NP-hard on tuples~$\I$ of open intervals even if (i) $G_\I$ is complete, (ii) the length of each interval is equal to its weight and (iii) all lengths and weights are positive integers.
\end{restatable}
\begin{proof}[Proof Sketch]
    We use the same construction (without $I_\ell$ and $I_r$) and assume that the intervals are open.
    Given a valid partition, we obtain a distance vector $D$ by moving the corresponding intervals consecutively on the two sides of $I_c$.
    Since the intervals are open, $G_{\I+D}\in\edgeless$, and a direct calculation yields $\onenorm{w\cdot D}=T$.
    In the other direction, let $D$ be a distance vector such that $G_{\I+D}\in\edgeless$ and $\onenorm{w\cdot D}\le T$.
    The intervals $I_1,\ldots,I_n$ are completely moved to the left or right of $I_c$.
    Moving these intervals towards the origin decreases the cost, so we may place them consecutively against $I_c$.
    Therefore by \Cref{lem:wequall_both_sides_cost} we obtain a cost of at least $T+(S_r-L)^2$ and $S_r=L$ since $\onenorm{w\cdot D}\le T$.
    Finally, all intervals contain the origin by construction, so their intersection graph is complete.
\end{proof}
Since every complete graph has exactly one maximal clique, \Cref{thm:wequall_edgeless_nphard} directly implies the following result.
\begin{corollary}
    \label{cor:edgeless_maximal_cliques_paranphard}
    \gged[\edgeless] on weighted open intervals is para-\NP-hard parameterised by the number $k$ of maximal cliques of $G_\I$ even when $k=1$.
\end{corollary}

Finally, we show that \gged[{\kconnected[1]}] on weighted unit disks is strongly \NP-hard by reducing from {\planarmonothreesat}.
Given a CNF formula~$\phi$ of monotone clauses (i.e., with only positive or only negative literals) of length at most~3 and whose incidence graph is planar, {\planarmonothreesat} asks whether $\phi$ is satisfiable.

\begin{restatable}[\restateref{thm:snph_wdisk}]{theorem}{snphwdisk}
    \label{thm:snph_wdisk}
    \gged[{\kconnected[1]}] is strongly \NP-hard even when restricted to collections $\D$ of weighted unit disks for which (i) $G_\D$ is a forest, (ii) centres have integer coordinates and (iii) the number of distinct weights is two.
\end{restatable}
\begin{proof}[Proof Sketch]
    We reduce from {\planarmonothreesat}.
    Using the rectilinear planar embedding of the given instance, we construct for each variable $x_i$ a gadget $\X_i$ consisting of a skeleton of heavy disks and several light disks.
    A (counter)clockwise rotation is a movement in which every light disk of $\X_i$ moves by a distance of $3$. Either rotation connects the whole variable gadget.
    We interpret the two rotations as the two truth values of $x_i$. \Cref{fig:variablegadget_example_conf} illustrates the rotation for the true assignment, where $\kappa$ is the maximum number of positive or negative occurrences of~$x_i$.
    The light disks $x_{i,j}$ and $\overline{x}_{i,j}$, $j\in[\kappa]$, correspond to the occurrences of~$x_i$ in positive and negative clauses, respectively.

    \begin{figure}[pbt]
        \centering
        \includegraphics[scale=1,page=8]{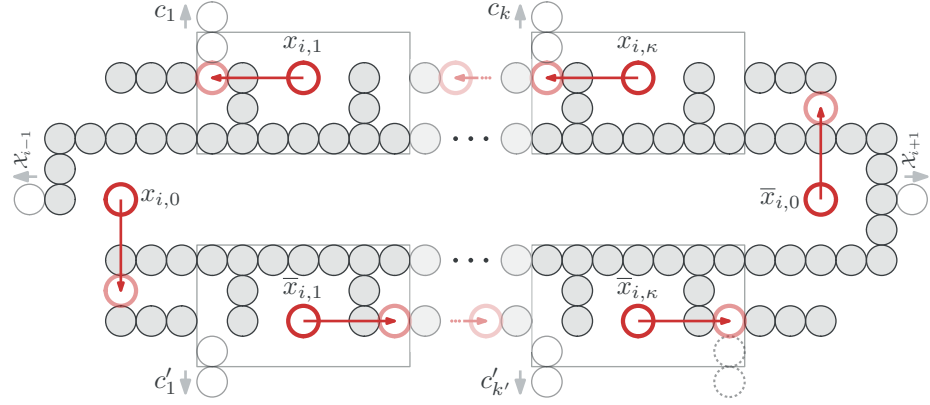}%
        \caption{Variable gadget $\X_i$ for $x_i$. Each positive or negative occurrence of $x_i$ in a clause has a corresponding light disk.
            The grey-shaded black disks form the skeleton for $\X_i$ and the dashed grey disks indicate an unused slot for a negative clause.
            The arrows and bold light-red disks show the rotation encoding the true assignment of $x_i$, which connects $\X_i$ to every clause containing $x_i$ positively.
        }
        \label{fig:variablegadget_example_conf}
    \end{figure}

    The variable gadgets are connected by chains of heavy disks. We represent each clause by a claw-shaped component of heavy disks, almost connected to the corresponding variable gadgets. A clause component is connected to a variable gadget exactly when the rotation satisfies the corresponding literal.

    Let $m$ be the total number of light disks.
    We give each light disk weight~$1$, every heavy disk weight~$W=12m$ and set the cost threshold to~$3m$.
    A satisfying assignment yields a connected intersection graph by applying the corresponding rotation to every $\X_i$, for a total cost of $3m$.
    In the other direction, any connected solution of cost at most~$3m$ cannot significantly move a heavy disk and the light disks connect each variable gadget only by one of the two rotations.
    These rotations define a truth assignment satisfying every clause, since each clause component must connect to a variable gadget.

    The construction has polynomial size, its initial intersection graph is a forest, all disk centres are integral, and the only weights used are $1$ and $W$.
\end{proof}

\section{Open Problems and Further Research}
\label{sec:open}

In this paper, we have shown new lower and upper bounds on the
(time) complexity of several graph editing problems for intersection
graphs of intervals and weighted unit disks.  Many variants of these
problems, however, are still open.
\begin{itemize}
    \item Can we solve \gged[{\kconnected}] on unit intervals in $O(n^2)$ time when $k\ge 2$?
    \item \gged[{\kconnected[1]}] is \NP-hard on
          weighted unit disks (with only two different weights).
          The authors of~\cite{Anari2016} presented an $O(n)$-factor approximation algorithm for the unweighted case.
          We conjecture this case remains \NP-hard.
    \item What is the complexity of
          \gged[{\kconnected[1]}] on weighted unit intervals?
    \item What is the complexity of \gged[{\kconnected[1]}] on a closed cycle?
    \item Is \gged[{\kconnected[1]}] {\FPT}
          on unweighted unit disks parameterised by the number of disks moved or on
          unweighted disks parameterised by the number of different radii?
    \item What is the complexity of \gged[\kclique] on unit disks?
\end{itemize}

\begin{credits}
    \subsubsection{\ackname} This work was supported by the Japan Society for the Promotion of Science (JSPS) KAKENHI Grant Numbers JP22H00513, JP24K02898, JP25K03077 and JP26K23808, and by the Japan Science and Technology Agency (JST) CRONOS Grant Number JPMJCS24K2.

    \subsubsection{\discintname} The authors have no competing interests to declare that are relevant to the content of this article.
\end{credits}

\bibliographystyle{splncs04}
\bibliography{bibliography.bib}

\arxiv{
    \newpage
    \appendix

    \section{Omitted Material for \texorpdfstring{\gged[{\kclique}]}{GGED(Pi-kclique)}}

    \kcliquegivenklevel* \label{claim:k_clique_given_k_level*}

    \begin{proof}
        \renewcommand{\qedsymbol}{$\triangle$}
        We describe an algorithm that computes the optimal point~$x^\star$ and a set $S^\star$
        of indices of intervals that need to be moved to~$x^\star$ such that the resulting intersection graph contains a $k$-clique.
        The algorithm considers the x-coordinates of vertices of~$L_k$ and the endpoints of the intervals as event points and traverses them from left to right. It maintains
        a candidate set of indices of intervals~$S$ that need to be moved, and its objective function $\sum_{i \in S} f_i$
        using two variables~$m$ and~$t$ that store the sum of slopes and the sum of y-intercepts of the functions~$f_i$ for $i\in S$, respectively.
        Note that~$x^\star$ may only change at event points
        since between two consecutive event points, the functions indexed by $S$ are fixed and thus our objective~$\sum_{i \in S} f_i$ is linear.
        Further, note that~$S$ can also only change at the x-coordinate of a vertex of~$L_k$ by the definition of~$L_k$, and
        we can assume that $x^\star$ is between the leftmost endpoint and the rightmost endpoint in $\I$.
        Additionally, $m$ and $t$ may only change if a new function is added to~$S$ or an endpoint of an interval is reached.
        Consequently, it suffices to process the event points in the range spanned by the leftmost and rightmost endpoints of $\I$.

        To this end, let~$a$ be the leftmost endpoint in~$\I$.
        We initialise the candidate set~$S$ by evaluating~$f_i$ at~$a$ for all~$i \in [n]$ to determine
        $S_a$ and set $S = S_a$. Using these indices, we initialise~$m$ and~$t$.
        Lastly, we set~$x^\star$ to~$a$.

        If an event point~$v$ corresponds to an endpoint of an interval~$I_j$ with $j\in S$, we update~$m$ and~$t$:
        If~$v$ is a left endpoint, we remove the negative slope~$-w_j$ from~$m$. Otherwise,~$v$ is a right endpoint and we add~$w_j$ to~$m$. Thus, in both cases, we set $m=m+w_j$.
        We can update~$t$ in a similar fashion by setting $t=t-w_ja_j$ and $t=t-w_jb_j$ for left and right endpoints, respectively.
        Finally, we evaluate the objective function and update~$x^\star$ if necessary.

        If an event point $x$ is the x-coordinate of a vertex~$v_i$ of~$L_k$ to the right of~$a$, we do the following.
        We consider the (directed) line~$v_{i-1}v_i$ and the vertex~$v_{i+1}$ of three consecutive vertices in $L_k$. %
        We distinguish between two cases depending on whether~$v_{i+1}$ is to the left or to the right of~$v_{i-1}v_i$.
        If~$v_{i+1}$ is to the left of~$v_{i-1}v_i$, the cost function that was previously one of the~$k-1$ cheapest functions is now the $k$th cheapest function.
        In this case,~$S$ does not change, hence we only need to evaluate the function at $x$ using~$m$ and~$t$, and update~$x^\star$ if a new minimum is found.
        If~$v_{i+1}$ is to the right of~$v_{i-1}v_i$, then the $k$th cheapest function~$f_\ell$ is exchanged with a new function~$f_j$.
        Therefore, we remove~$\ell$ from~$S$ and add~$j$.
        Similarly, we subtract the current slope of~$f_\ell$ from $m$ and add the current slope of~$f_j$. We update~$t$ analogously using the current y-intercepts.
        Finally, we evaluate the total cost of the~$k$ cheapest functions using the updated slope~$m$ and y-intercept~$t$, and update~$x^\star$ if necessary.

        After processing all event points,~$x^\star$ is the optimal position to move intervals of~$\I$ such that the resulting intersection graph contains a $k$-clique.
        After the sweep, we recover $S^\star$ in $O(n)$ time by selecting the indices of the $k$ smallest values $f_i(x^\star)$.

        Observe that we can implement~$S$ using a binary array of size~$n$ where the $i$th entry is true if and only if interval~$I_i$ is currently in~$S$, making the implementation of the algorithm possible in constant time per iteration.
        Therefore, the total runtime is~$O(|L_k| + n)$.
    \end{proof}

    \section{Omitted Material for \texorpdfstring{\gged[{\kconnected}]}{GGED(Pi-k-conn)}}

    \circequaldualopt*\label{claim:circ_equal_dual_opt*}
    \begin{proof}
        \renewcommand{\qedsymbol}{$\triangle$}
        Let $f$ be a min-cost circulation for $N_k$.
        According to our construction, we set $\alpha_i = f(s,v_i)$, $\beta_i = f(v_i,t)$, $\lambda_i = f(v_{i+1},v_{i})$ and $\mu_i = f(v_i,v_{i+k})$ which satisfy~\eqref{eq:dual_k_connected_1} due to the flow conservation of $f$.
        However, \eqref{eq:dual_k_connected_2} might not be satisfied if
        $f(s,v_i) > 0$ and $f(v_i,t)>0$. In this case,
        let $\delta_i = \min{\{\alpha_i,\beta_i\}}$.
        By removing $\delta_i$ from $f(s,v_i)$, $f(v_i,t)$ and $f(t,s)$, we obtain a circulation with the same cost satisfying~\eqref{eq:dual_k_connected_2} since $p(s,v_i)+p(v_i,t)+p(t,s)=c(I_i)-c(I_i)+0=0$.
        We redefine $\alpha_i$ and $\beta_i$ using the corresponding reduced values.
        Therefore, the resulting variables are feasible for~\eqref{eq:dual_k_connected} and their objective value is $p(f)$.

        Conversely, let $(\alpha,\beta,\lambda,\mu)$ be an optimal solution of~\eqref{eq:dual_k_connected}. We assign the corresponding values to the arcs of $N_k$ by setting $f(s,v_i)=\alpha_i$, $f(v_i,t)=\beta_i$, $f(v_{i+1},v_i)=\lambda_i$ and
        $f(v_i,v_{i+k})=\mu_i$.
        By construction, flow is conserved at every vertex $v_i$.
        Moreover,~\eqref{eq:dual_k_connected_2} ensures that capacities are not exceeded, and \eqref{eq:dual_k_connected_3} ensures that the flow is nonnegative.
        The total flow from $s$ is $\sum_{i=1}^n \alpha_i$ and the total flow to $t$ is $\sum_{i=1}^n \beta_i$.
        Summing up~\eqref{eq:dual_k_connected_1}, we obtain:
        \begin{align*}
            0 = \sum_{i=1}^n-\alpha_i+\beta_i+\sum_{i=1}^n(\lambda_{i-1}-\lambda_i)+ \sum_{i=1}^n(\mu_i-\mu_{i-k})
            = -\sum_{i=1}^n \alpha_i+ \sum_{i=1}^n \beta_i,
        \end{align*}
        due to telescoping sums.
        Thus, we set $f(t,s)= \sum_{i=1}^n \alpha_i = \sum_{i=1}^n \beta_i$ and $f$ is a circulation of $N_k$.
        Lastly, due to our construction of $N_k$, the cost of $f$ corresponds to the objective function value of~\eqref{eq:dual_k_connected}.
    \end{proof}

    \mincostcirculation*\label{claim:min_cost_circulation*}
    \begin{proof}
        \renewcommand{\qedsymbol}{$\triangle$}
        We show how to calculate the circulation $f$.
        We modify~$N_k$ so that we can interpret a circulation as an $s$--$t$ flow.
        First, note that any flow in $N_k$ satisfies $\sum_{i=1}^n f(s,v_i) \le n$ and any flow going through the cycles connecting $v_1,\ldots,v_n$ has positive cost.
        Thus, the value of an optimal circulation is bounded by $n$ and we safely set all infinite capacities to $n$ since removing cycles in $v_1,\ldots,v_n$ decreases the cost.
        Any circulation in $N_k$ is completed at the arc $(t,s)$.
        We remove $(t,s)$ and add the arc $(s,t)$ with cost $p(s,t)=0$ and capacity $u(s,t)=n$.
        Let $N'_k$ be the modified version of $N_k$.
        A circulation is equivalent to an $s$--$t$ flow of value $n$ in the modified network.
        Assume that we have $f(t,s) = q$.
        This yields an $s$--$t$ flow of value $q$ in $N'_k$.
        Sending $n-q$ units of flow through $(s,t)$ results in a flow of total value $n$ and the same cost.
        Conversely, assume $f'$ is an $s$--$t$ flow of value $n$ in $N'_k$ with $f'(s,t) =n - q$.
        Then assigning $f(t,s) = q$ and removing $(s,t)$ produces a circulation of the same cost in $N_k$.
        Consequently, a min-cost circulation in $N_k$ is equivalent to a min-cost $s$--$t$ flow of value $n$ in $N'_k$.

        Angelov, Harb, Kannan, and Wang~\cite{Angelov2006} argue that translating arc costs to make them nonnegative allows an $s$--$t$ flow to be calculated in $O(mn+n^2\log n)$ time
        by applying Dijkstra's algorithm at most $2n+1$ times, where $m$ is the number of arcs.
        For our network, each (indirect) augmenting path $s$--$t$ saturates a pair of arcs $(s,v_i)$ and $(v_j,t)$ since we have unit capacities.
        Moreover, there are no negative cost cycles.
        Thus we have at most $n$ augmentations, so we apply Dijkstra's algorithm at most $n$ times.
        This yields a running time of $O(n^2\log n)$ for $N'_k$.
    \end{proof}

    \section{Omitted Material for the NP-hardness Proofs}
    \label{sec_apx:hardness_proofs}
    This section gives all the details of the \NP-hardness results of the paper.
    \subsection{\texorpdfstring{\gged[{\kconnected[1]}]}{GGED(Pi-conn)} is strongly NP-hard on Intervals of Arbitrary Length}

    We show that \gged[{\kconnected[1]}] is \NP-hard via a reduction from {\threepartition}.
    We adapt the reduction used to prove that
    \gged[\edgeless] is strongly \NP-hard on unweighted intervals~\cite{HonoratoDroguett2026}, where $\edgeless$ is the class of graphs with no edges.
    A similar reduction for interval coverage appears in~\cite{Czyzowicz2010}.
    \defdecproblem{\threepartition}%
    {A multiset $A=\{a_1,\dots,a_{3m}\}$ of positive integers and a bound
        $L\in \mathbb{Z}^+$ such that $\sum_{i \in [3m]} a_i = mL$ and, for
        every $i \in [3m]$, it holds that $L/4 <a_i < L/2$.}%
    {Decide whether $A$ can be partitioned into $m$ multisets
        $A_1,\ldots,A_m$ of size~3 such that, for every $j \in [m]$, the
        elements in~$A_j$ sum up to~$L$.}

    \snphintv*\label{thm:snph_intv*}
    \begin{proof}
        Given an instance $(A,L)$ of {\threepartition}, we set $T = 3m(3mL+L)/4$.
        We construct an instance~\I of \gged[{\kconnected[1]}] and show that $(A,L)$ admits a valid partition
        if and only if there is a distance vector~$D$ such that $G_{\I+D} \in \kconnected[1]$ and $\|D\| < T$.
        We construct an instance $\I$ consisting of two tuples $\A$ and $\C$ as follows.
        \Cref{fig:redexample_intv1conn} illustrates the reduction.
        \begin{itemize}
            \item The tuple $\A = (I_i)_{i\in [3m]}$ contains, for each
                  $i \in [3m]$ and $a_i \in A$, an interval~$I_i$ with
                  $\len{I_i} = a_i$ and $c(I_i) = -a_i/2$.
            \item $\C$ consists of the \emph{component} tuples $\C_0,\ldots,\C_m$ of $\I$, where
                  \begin{itemize}
                      \item[--] for $i\in [m-1]$, $\C_i = (B^i_1,\ldots, B^i_j)$ is a tuple of intervals, an \emph{inner component}, such that $j = (8T+1)(\lfloor T\rfloor+1)$. For $i \in [m-1] $ and $h \in [j]$, we set $\len{B^i_h} = L/(2j)$ and $c(B^i_h) = (3i-1)L/2 + (2h-1)L/(4j)$.
                      \item[--] for $i \in \set{0,m}$, $\C_i = (B^i_1,\ldots, B^i_j)$ is a tuple of intervals, an \emph{outer component}, such that $j = (8T+1)(\lfloor T\rfloor+1+ \max A)$. We set $\len{B^i_h} = 1/(8T+1)$, $c(B^0_h) = (2h-2j-1)/(2(8T+1))$, and $c(B^m_h) = (3m-1)L/2 + (2h-1)/(2(8T+1))$, for $h \in [j]$.
                  \end{itemize}
        \end{itemize}

        We call the empty spaces between components of $\I$ \emph{gaps}.
        Each component~$\C_i$ is a path in the intersection graph.
        For $i\in [m-1]$, the union of~$\C_i$ has total length~$L/2$.
        Intuitively, we want to ensure that the intervals in~$\A$ are used to bridge the gaps of~$\C$.
        The components consist of many short intervals instead of a single interval since a single interval could be shifted by~$\delta$ into a gap at cost~$\delta$, whereas covering a part of length~$\delta$ of a gap with intervals of~$\C$, each of length at most $1/(8T+1)$, requires moving at least $\delta(8T+1)$ of them.
        Given the above instance $\I$ for an instance $(A,L)$ of {\threepartition}, we show that using
        the intervals in~$\A$ to make~$G_\I$ connected is not too costly.
        \begin{claim}\label{claim:tmd_1conn_intv}
            If $(A, L)$ is a yes-instance of \threepartition, then
            the total moving distance required to bridge every gap with three intervals from $\A$ is less than $T$.
        \end{claim}
        \begin{proof}
            \renewcommand{\qedsymbol}{$\triangle$}
            Let $I_i,I_j,I_k$ be intervals in $\A$ corresponding to elements $a_i,a_j,a_k$ in $A$ such that $\len{I_i}+\len{I_j}+\len{I_k} = L$.
            Assume we move intervals $I_i,I_j,I_k$ to bridge the gap between $\C_{\ell-1}$ and $\C_\ell$.
            We divide the total moving distance for the three intervals into the moving distance to align the right endpoints of $I_i, I_j, I_k$ with the rightmost endpoint of $\C_{\ell-1}$ and the moving distance of $I_i, I_j, I_k$ to fill the gap.
            The first part requires a total moving distance of $3(3L/2)(\ell-1)$.
            Without loss of generality, we fill the gap from left to right by moving first $I_i$, then $I_j$, and finally $I_k$.
            The total moving distance of this movement is $3a_i+2a_j+a_k < 6L/2$, since every element in $A$ is smaller than $L/2$.
            For the $m$ gaps, we obtain:
            \begin{align*}
                \sum_{\ell\in [m]} & \left(\frac{9L(\ell-1)}{2}\right) + \frac{6mL}{2} = \frac{9L}{2}\left(\frac{m(m+1)}{2}-m\right) + \frac{6mL}{2} \\
                                   & = \frac{9L}{2}\left(\frac{m^2-m}{2}\right) +\frac{6mL}{2} = \frac{9m^2L+3mL}{4} = T.
            \end{align*}
            This concludes the proof.
        \end{proof}
        Using this transformed instance we can decide \threepartition.
        \begin{claim}\label{claim:1_connected_arbitrary_length_equivalence}
            A pair $(A,L)$ is a yes-instance of {\threepartition} if and only if there is a distance vector $D$ such that $G_{\I+D} \in \kconnected[1]$ and $\|D\| < T$.
        \end{claim}
        \begin{proof}
            \renewcommand{\qedsymbol}{$\triangle$}
            Assume that $(A,L)$ is a yes-instance of {\threepartition}, and let
            $A_1,\dots,A_m$ be a partition of $A$ such that $\sum_{a\in A_i} a = L$ for every $i\in [m]$.
            Let $\I_1,\ldots,\I_m$ be the partition of $\A$ corresponding to $A_1,\dots,A_m$.
            For each $i \in [m]$, we fill the gap between $\C_{i-1}$ and $\C_i$ using $\I_i$.
            Since $\sum_{I\in \I_i} \len{I} = L$ for every $i\in [m]$, it follows that $\C_{i-1}$ and $\C_i$ are connected.
            Moreover, the total moving distance is less than $T$ by \Cref{claim:tmd_1conn_intv}.
            Hence, there is a distance vector~$D$ such that $G_{\I+D} \in \kconnected[1]$ and $\|D\| < T$.

            Now assume that there is a distance vector~$D$ such that $G_{\I+D} \in \kconnected[1]$ and $\|D\| < T$.
            For each $i\in[m]$, let $\A_i$ be the intervals of $\A$ that, after movement, cover part of the original gap between $\C_{i-1}$ and $\C_i$.
            Since consecutive gaps are separated by $L/2$ and every interval of $\A$ has length less than $L/2$, no interval appears in two such sets.
            Hence the sets $\A_i$ are pairwise disjoint.
            We show that $\sum_{I\in \A_i}\len{I} \ge L$ for every $i\in[m]$.
            Suppose to the contrary that $\sum_{I\in \A_i}\len{I} < L$ for some $i\in[m]$.
            Since all lengths of intervals in~$\A$ are positive integers, we have $\sum_{I\in \A_i}\len{I} \le L-1$.
            The total length of each outer component is more than $T$ to the left and right of the gaps, so connectivity and $\onenorm{D}<T$ imply that all original gaps are covered.
            By construction, the $i$th gap has length $L$ and, independently of the order, the intervals of $\A_i$ cover at most $L-1$ of it.
            Thus intervals of $\C$ must fill a total length of at least $1$.
            Consider now excluding a portion of length $1/4$ at each end of the original gap.
            The length still to fill is at least $1/2$, and every interval of $\C$ covering part of it must be moved by at least $1/4$ (imagine intervals `crossing' the excluded portion).
            Since $T\ge 3L$, the length of inner and outer intervals in $\C$ is at most $1/(8T+1)$.
            Consequently, we need to move at least $(8T+1)/2$ intervals, giving a total cost of $1/4\cdot(8T+1)/2> T$.
            Hence, such an $\A_i$ contradicts $\|D\| < T$ and we conclude that $\sum_{I\in \A_i}\len{I} \ge L$ for every $i\in[m]$.
            Since the sets $\A_i$ are pairwise disjoint and $\sum_{I\in \A}\len{I}=mL$, it follows that $\sum_{I\in \A_i}\len{I} = L$ for every $i \in [m]$.
            Finally, as $L/4<\len{I}<L/2$ for every $I\in\A$, each set~$\A_i$ contains exactly three intervals, implying that $\A_1,\ldots,\A_m$ give a partition of~$A$.
        \end{proof}

        Note that by the strong \NP-hardness of {\threepartition}, we may restrict to instances with $L$ (and thus $T$) polynomially bounded by the input size. Hence, we can construct $\I$ in polynomial time from $(A,L)$.
        Lastly, all values used in the reduction are rational numbers.
    \end{proof}

    \subsection{Weak NP-hardness of \texorpdfstring{\gged[{\kconnected[1]}]}{GGED(Pi-conn)} on a Special Case of Intervals of Arbitrary Length}

    We consider the special case of weighted intervals in which lengths equal weights.
    We show that \gged[{\kconnected[1]}] is weakly \NP-hard on these instances.
    Afterwards, we adapt the reduction to prove the same result for \gged[\edgeless].
    The authors of~\cite{HonoratoDroguett2026} stated this case as open.
    The two reductions use the same construction idea and the same cost bound.
    For connectivity, we assume that intervals are closed, whereas for independence we assume that intervals are open.
    We first establish the common skeleton and then prove the two results separately.
    The reduction is from {\partition}~\cite{Garey1979}:
    \defproblem
    {{\partition}}
    {A multiset $A=\{a_1,\ldots,a_n\}$ of $n$ positive integers with $\sum_{i=1}^n a_i=2L$.}
    {A subset $A'\subseteq A$ such that $\sum_{a\in A'}a = \sum_{a\in A\setminus A'}a=L$.}
    {Output}

    We construct a tuple of intervals $(I_1,\ldots,I_n,I_c)$ using $A$ as follows.
    For every $i\in[n]$, we introduce an interval $I_i$ with $c(I_i)=0$ and $\len{I_i} =w_i=a_i$.
    We also introduce an extra interval $I_c$ with $c(I_c)=0$ and $w_c=\len{I_c}=4L+1$.
    Lastly, we set $T=w_cL+L^2$.

    The intuitive idea of the two reductions is as follows.
    The interval $I_c$ acts as a {\em divider} for the two parts of the partition.
    The intervals are placed either to its left or to its right.
    The cost threshold forces the total length on each side to be~$L$.
    For connectivity, we will place two heavy outer intervals which delimit two \emph{gaps} of length~$L$, and the intervals connect the divider to the outer intervals (recall that the intervals in this case are closed).
    For the edgeless case, we omit the outer intervals, and we can pack the open item intervals consecutively on the two sides without intersecting.
    We start by noting that $I_c$ must contain the origin in every solution of total weighted moving distance of at most $T$.

    \begin{observation}
        \label{obs:wequall_divider_origin}
        If $\onenorm{w \cdot D}\le T$, then $I_c+d_c$ contains the origin.
    \end{observation}
    \begin{proof}
        Since $c(I_c)=0$, its moving distance is $\abs{d_c}$, and its weighted moving distance is $w_c\abs{d_c}$.
        Thus $\abs{d_c}\le T/w_c = L+L^2/w_c < \len{I_c}/2 = w_c/2$.
        Therefore the centre of $I_c+d_c$ is at distance less than half its length from the origin.
    \end{proof}

    We next show, for a subset of intervals, the total weighted moving distance of moving the intervals to one side of $I_c$.
    We state the lemma for the right side.
    However, we can reflect the instance through the origin and argue the same for the left side.
    The authors of~\cite{Lenstra1977} previously used a reduction based on the same idea for the two-machine weighted-completion-time scheduling problem.
    In our setting, however, we must consider the cost of moving the interval $I_c$.

    \begin{lemma}
        \label{lem:wequall_one_side_cost}
        Suppose the centre of $I_c$ is moved to point~$y$ and $I_c$ contains the origin.
        If the intervals of a set $\I_r\subseteq\{I_1,\ldots,I_n\}$ are placed consecutively with no gaps to the right of $I_c$ and $S_r=\sum_{I_i\in \I_r}\len{I_i}=\sum_{I_i\in \I_r}a_i$, then the total weighted moving distance of the intervals in $\I_r$ is $(\len{I_c}/2+y)S_r+S_r^2/2$.
    \end{lemma}
    \begin{proof}
        If $\I_r=\emptyset$, then $S_r=0$ and the claim follows.
        Thus assume $\I_r\neq\emptyset$ and relabel the intervals so that $\I_r=\set{I_1,\ldots,I_k}$ with indices following the final left-to-right order.
        Since the intervals are consecutive and lie to the right of the moved $I_c$, the final centre of $I_j$ is
        $\len{I_c}/2+y+\sum^{j-1}_{h=1}\len{I_h}+\len{I_j}/2$.
        This point is positive because the moved $I_c$ contains the origin.
        Since $c(I_j)=0$ initially, its moving distance corresponds to the final centre of $I_j$.
        Multiplying this distance by $w_j$ yields its weighted moving distance.
        The total weighted moving distance is
        \begin{align*}
            \sum_{j=1}^k & w_j\left(\frac{\len{I_c}}{2}+y+\sum^{j-1}_{h=1}\len{I_h}+\frac{\len{I_j}}{2}\right)                        \\
                         & = \left(\frac{\len{I_c}}{2}+y\right)S_r+\sum_{j=1}^k\sum^{j-1}_{h=1}a_h a_j+\frac{1}{2}\sum_{j=1}^k a_j^2. %
        \end{align*}
        Expanding $(\sum_{j=1}^k a_j)^2$ yields $\sum_{j=1}^k a_j^2 + 2 \sum_{j=1}^k\sum_{h=1}^{j-1} a_j a_h$. Therefore the above value equals $(\len{I_c}/2+y)S_r+(1/2)(\sum_{j=1}^k a_j)^2 = (\len{I_c}/2+y)S_r+S_r^2/2.$
    \end{proof}

    We remark here that \Cref{lem:wequall_one_side_cost} does not require any particular order for the moved intervals.
    Each interval $I_i$ contributes $w_i(\len{I_c}/2+y+\len{I_i}/2)$ to the cost, plus $w_i\len{I_j}$ for every interval $I_j$ placed between $I_c$ and $I_i$.
    If $I_i$ is between $I_c$ and $I_j$, then the value $w_j\len{I_i}=a_i a_j$ appears in the weighted moving distance of $I_j$.
    This is the same property noted by Jansen and Kahler~\cite{Jansen2023} for jobs whose processing times equal their weights.

    For a solution of our constructed instance, let the intervals placed to the right of $I_c$ have total length $S_r$, and let the intervals placed to the left of $I_c$ have total length $S_\ell=2L-S_r$.
    We show the following consequence of \Cref{lem:wequall_one_side_cost} and use it in both reductions.
    \wequallbothsidescost*\label{lem:wequall_both_sides_cost*}
    \begin{proof}
        Let $D$ be the distance vector where the centre of $I_c$ is moved to point~$y$.
        Applying \Cref{lem:wequall_one_side_cost} to the right and left sides, we obtain
        \begin{align*}
            \onenorm{w \cdot D}
             & \ge w_c\abs{y}
            + \left(\frac{\len{I_c}}{2}-y\right)S_\ell+\frac{S_\ell^2}{2}
            + \left(\frac{\len{I_c}}{2}+y\right)S_r+\frac{S_r^2}{2}                \\
             & = w_c\abs{y}
            + \left(\frac{\len{I_c}}{2}-y\right)(2L-S_r)+\frac{(2L-S_r)^2}{2}      \\
             & \phantom{=} + \left(\frac{\len{I_c}}{2}+y\right)S_r+\frac{S_r^2}{2} \\
             & = w_c\abs{y}+\len{I_c}L+L^2+(S_r-L)^2+2(S_r-L)y.
        \end{align*}
        Since $\abs{S_r-L}\le L$ and $w_c>2L$, it follows that
        $w_c\abs{y}+2(S_r-L)y\ge (w_c-2L)\abs{y}\ge 0$.
        Therefore $\onenorm{w \cdot D}\ge \len{I_c}L+L^2+(S_r-L)^2=T+(S_r-L)^2$.
    \end{proof}

    \paragraph{Connectivity.}

    We start by showing the reduction for \gged[{\kconnected[1]}].
    In this case, we add two more intervals denoted by $I_\ell$ and $I_r$ (see again~\Cref{fig:partition_example}), one on each side of $I_c$ and separated from it by a gap of length $L$.
    Intuitively, we connect these two intervals to $I_c$ using $I_{1},\ldots,I_n$.
    We show that this can be done with total weighted moving distance of at most $T$ if and only if $A$ is a yes-instance of {\partition}.

    Formally, let $I_\ell$ and $I_r$ be two intervals whose lengths and weights equal $2T+1$, and set $c(I_\ell) = -(\len{I_c}/2 + L + (2T+1)/2)$ and $c(I_r) = -c(I_\ell)$.
    Then $I_\ell$ and $I_r$ correspond to the intervals described above, and the resulting instance is $\I = (I_\ell,I_1,\ldots,I_n,I_c,I_r)$.

    \wequallconnnphard*\label{thm:wequall_conn_nphard*}
    \begin{proof}
        Let $A$ be an instance as described above.
        \begin{claim}\label{claim:wequall_partition_iff_conn}
            An instance $A$ of {\partition} is a yes-instance if and only if there is a distance vector $D$ such that $G_{\I+D}\in\kconnected[1]$ and $\onenorm{w \cdot D}\le T$.
        \end{claim}
        \begin{proof}
            \renewcommand{\qedsymbol}{$\triangle$}
            First, assume that $A$ is a yes-instance of {\partition}.
            Let $A'\subseteq A$ be such that $\sum_{a\in A'}a=L$.
            Let $\I'=\set{I'_1,\ldots,I'_k}= \set{I_i \in \I\colon a_i \in A'}$ be the intervals corresponding to elements of $A'$ labelled in arbitrary order.
            We place the intervals of $\I'$ to the right of $I_c$ such that $c(I'_{i+1})+d'_{i+1} - (c(I'_{i})+d'_{i}) = (\len{I'_{i+1}}+\len{I'_{i}})/2$ and the left endpoint of $I'_1+d'_1$ intersects with the right endpoint of $I_c$.
            We also move the intervals $\set{I_1,\ldots,I_n} \setminus \I'$ in the same manner to the left of $I_c$.
            Lastly, we set $d_c = d_\ell = d_r = 0$.
            Let $D$ be the obtained distance vector.
            Since $\sum_{i=1}^k \len{I'_i} = L$, the interval $I'_k$ intersects exactly the left endpoint of $I_r$.
            Thus there is a path $v_c,v'_1,\ldots,v'_{k},v_{r}$ in $G_{\I+D}$.
            Similarly, $I_c$, $I_\ell$, and the remaining intervals form a path from $v_c$ to $v_\ell$.
            Hence $G_{\I+D}$ is in $\kconnected[1]$.
            The length of the intervals moved to the left and right is $L$, thus the total weighted moving distance is equal to $T$ by \Cref{lem:wequall_one_side_cost}.
            Therefore there is a distance vector~$D$ such that $G_{\I+D}\in \kconnected[1]$ and $\onenorm{w\cdot D}\le T$.

            In the other direction, assume that such a distance vector $D$ exists.
            We first show that $D$ describes a movement in which all intervals $I_1,\ldots,I_n$ are moved.
            In other words, we need all $I_1,\ldots,I_n$ to connect $I_\ell, I_c, I_r$.
            Assume without loss of generality that the centre of $I_c$ is moved to $y=d_c$ such that $y\ge 0$.
            For the intervals $I_\ell$ and $I_r$, we must have $(2T+1)(\abs{d_\ell}+\abs{d_r})\le T$ and thus $\abs{d_\ell}+\abs{d_r}<1/2$.
            Consequently, the sum of the lengths of the two gaps adjacent to $I_c$ is $\max\{0,L+y-d_\ell\}+\max\{0,L-y+d_r\} > 2L-1/2$.
            Since $\sum_{i=1}^n \len{I_i} = 2L$ and $\len{I_i} \in \mathbb{N}_{>0}$ for all $i\in [n]$, we need all intervals to cover both gaps (otherwise the total length is at most $2L-1$).
            Moreover, the vector $D$ must satisfy $\abs{d_\ell}+\abs{d_r}<1/2$ and by \Cref{obs:wequall_divider_origin}, $I_c+d_c$ contains the origin.
            Hence if $\onenorm{w \cdot D}\le T$, the intervals $I_{1}+d_{1},\ldots,I_{n}+d_n$ cover the gaps between $I_\ell, I_c, I_r$. Moreover, none of the intervals $I_1,\ldots,I_n$ is contained in $I_c+d_c$, and they form two chains between $I_{\ell}$, $I_c$ and $I_r$.
            We show that these two chains have length equal to $L$.

            Assume $\I_r$ is the set of intervals covering the right gap and $S_r$ is their total length.
            Thus the total length of the intervals covering the left gap is $S_\ell = 2L-S_r$.
            The signed lengths of the left and right gaps in $\I+D$ are $L +y - d_\ell$ and $L - y + d_r$, respectively.
            Let $\lambda_\ell = S_\ell - (L +y - d_\ell)$ and $\lambda_r = S_r - (L - y + d_r)$ be the differences between the chain lengths and the corresponding signed gap lengths.
            Since the intervals cover the gaps, we have that $\lambda_\ell, \lambda_r \ge 0$.
            Here we recall \Cref{lem:wequall_both_sides_cost}.
            The movement described by the lemma places each interval consecutively to the left and right of $I_c$ and the total cost of such a movement is at least $T+(S_r-L)^2$.
            Since moving an interval closer to the origin decreases its cost, the optimal movement to intersect $I_{\ell}+d_\ell$ and $I_r+d_r$ must consist of consecutive touching intervals from the right endpoint of $I_\ell+d_\ell$ and the left endpoint of $I_r+d_r$ until $I_c$ is intersected.
            This is equivalent to moving the intervals as described in \Cref{lem:wequall_both_sides_cost} and then shifting the intervals by $\lambda_\ell$ and $\lambda_r$ respectively, towards the origin.
            Hence the total weighted moving distance of $I_{1},\ldots,I_n,I_c$ in $D$ is at least $T+(S_r-L)^2 - S_\ell\lambda_\ell - S_r\lambda_r$.
            We further bound this value.
            First, since $-2L \le -S_\ell,-S_r$, we obtain $-S_\ell\lambda_\ell - S_r\lambda_r \ge -2L(\lambda_\ell + \lambda_r)$ and $-2L(\lambda_\ell + \lambda_r)= -2L(d_\ell -d_r)$ by the definition of $\lambda_\ell$ and $\lambda_r$.
            Adding the weighted moving distance of $I_\ell$ and $I_r$, we bound the cost $\onenorm{w\cdot D}$ as follows:
            \begin{align*}
                \onenorm{w\cdot D} & \ge T+(S_r-L)^2 - S_\ell\lambda_\ell - S_r\lambda_r + (2T+1)(\abs{d_\ell}+ \abs{d_r}) \\
                                   & \ge T+(S_r-L)^2 -2L(d_\ell -d_r) + (2T+1)(\abs{d_\ell}+ \abs{d_r})                    \\
                                   & \ge T+(S_r-L)^2 -2L(\abs{d_\ell}+ \abs{d_r}) + (2T+1)(\abs{d_\ell}+ \abs{d_r})        \\
                                   & \ge T+(S_r-L)^2 + (2T+1-2L)(\abs{d_\ell}+ \abs{d_r})                                  \\
                                   & \ge T+(S_r-L)^2,
            \end{align*}
            since $-(d_\ell - d_r) \ge -(\abs{d_\ell}+ \abs{d_r})$ and $2T+1-2L>0$.
            Hence $S_r = L$ since we have $\onenorm{w\cdot D} \le T$.
            Therefore the corresponding elements of $A$ for $\I_r$ form a partition of $A$.
            This completes the proof.
        \end{proof}
        The correctness of the reduction follows from \Cref{claim:wequall_partition_iff_conn} and the construction of $\I$ takes polynomial time since $\size{\I} = n+3$.
        Lastly, $G_\I$ has three connected components by construction and the definition of {\partition} guarantees the integrality of the instance.
    \end{proof}

    \paragraph{The Edgeless Case.}
    Lastly, we adapt the reduction to the edgeless case.
    \wequalledgelessnphard*\label{thm:wequall_edgeless_nphard*}
    \begin{proof}
        Let $A$ be an instance as described above.
        \begin{claim}\label{claim:wequall_partition_iff_edgeless}
            An instance $A$ of {\partition} is a yes-instance if and only if there is a distance vector $D$ such that $G_{\I+D}\in\edgeless$ and $\onenorm{w \cdot D}\le T$.
        \end{claim}
        \begin{proof}
            \renewcommand{\qedsymbol}{$\triangle$}
            First, assume that $A$ is a yes-instance of {\partition}.
            Let $A'\subseteq A$ be such that $\sum_{a\in A'}a=L$.
            We set $d_c = 0$ and let $\I'=\set{I'_1,\ldots,I'_k}= \set{I_i \in \I\colon a_i \in A'}$ be the intervals corresponding to elements of $A'$ labelled in arbitrary order.
            We place the intervals corresponding to the elements of $A'$ to the right of $I_c$ such that $c(I'_{i+1})+d'_{i+1} - (c(I'_{i})+d'_{i}) = (\len{I'_{i+1}}+\len{I'_{i}})/2$, and the remaining intervals in the same manner to the left of $I_c$.
            Hence $G_{\I+D}\in\edgeless$ since intervals are open.
            Moreover, since $y=0$ ($I_c$ is fixed) and $S_r=L$, the weighted moving distance of the intervals on the right is $\len{I_c}L/2+L^2/2$ by \Cref{lem:wequall_one_side_cost}.
            The intervals on the left have the same weighted moving distance.
            Consequently, we have $\onenorm{w \cdot D}=\len{I_c}L+L^2=w_cL+L^2=T$.

            Conversely, assume that there is a distance vector $D$ such that $G_{\I+D}\in\edgeless$ and $\onenorm{w \cdot D}\le T$.
            By \Cref{obs:wequall_divider_origin}, $I_c+d_c$ contains the origin.
            Let $y$ be the centre of $I_c+d_c$.
            Since $G_{\I+D} \in \edgeless$, each interval $I_i+d_i$, $i\in[n]$, lies completely to the left or completely to the right of $I_c+d_c$.

            Let $\I_r\subseteq\{I_1,\ldots,I_n\}$ be the set of intervals placed to the right of $I_c+d_c$, and let $S_r=\sum_{I_i\in \I_r}\len{I_i}=\sum_{I_i\in \I_r}a_i$.
            The intervals placed to the left of $I_c+d_c$ have total length $S_\ell=2L-S_r$.
            Since every item interval is initially centred at the origin and has positive weight, closing any gap by moving the intervals beyond it towards the origin decreases the total weighted moving distance.
            Hence, placing the intervals consecutively to the left and right of $I_c+d_c$ minimises the total weighted moving distance.
            Applying \Cref{lem:wequall_both_sides_cost}, we obtain $\onenorm{w \cdot D}\ge T+(S_r-L)^2$.
            Since $\onenorm{w \cdot D}\le T=w_cL+L^2$, we obtain $S_r=L$.
            Hence the elements represented by the intervals of $\I_r$ give a solution to the instance $A$ of {\partition}.
            This concludes the proof.
        \end{proof}
        Again, the correctness of the reduction follows from \Cref{claim:wequall_partition_iff_edgeless} and the construction of $\I$ takes polynomial time since $\size{\I} = n+1$.
        Lastly, $G_\I$ is complete by construction and the definition of {\partition} guarantees the integrality of the instance.
    \end{proof}

    \subsection{Strong NP-hardness of \texorpdfstring{\gged[{\kconnected[1]}]}{GGED(Pi-conn)} on Weighted Unit Disks}

    We show that \gged[{\kconnected[1]}] on weighted unit disks is strongly \NP-hard by reducing from {\planarmonothreesat}.
    Recall that a clause is \emph{monotone} if it contains only positive or only negative literals.
    We say that a clause is \emph{positive} or \emph{negative} in accordance with its literals.
    For a 3-SAT formula $\phi$, let $X_\phi$ be the set of variables in~$\phi$ and let $C_\phi$ be the set of clauses in~$\phi$.
    Then the \emph{incidence graph} $G_\phi = (X_\phi \cup C_\phi, E_\phi)$ of $\phi$ is the bipartite graph that has an edge $\{x,c\}$ whenever,  in~$\phi$, $x$ or $\overline{x}$ occurs in~$c$.

    \defdecproblem{\planarmonothreesat}%
    {A CNF formula $\phi$ given by a set $C_\phi$ of
        monotone clauses of length at most~3 over a set
        $X_\phi$ of variables such that the incidence graph~$G_\phi$ is planar.}%
    {Decide whether $\phi$ is satisfiable.}

    It is known that {\planarmonothreesat} is \NP-hard even if we insist
    that $G_\phi$ is drawn such that all variables are arranged on a
    horizontal line, positive clauses appear above the line, and negative
    clauses appear below the line~\cite{DeBerg2012}; see
    \Cref{fig:pm3sat_example}~(left).

    \snphwdisk*\label{thm:snph_wdisk*}
    \begin{proof}
        We construct gadgets following the embedding of~$G_\phi$ described above.
        \Cref{fig:pm3sat_example}~(right) illustrates the skeleton of the reduction.
        The core of our reduction is the variable gadget, in which some disconnected disks need to move.
        Our construction uses weighted disks of diameter~1.
        We use only two different weights.
        A \emph{light disk}
        is a disk with unit weight, and a \emph{heavy disk} is a disk with weight~$W$, which we will fix later.
        We will make $W$ large enough so that we can assume that heavy disks are not moved in a yes-instance of the problem.
        In the following figures, we illustrate heavy disks using black thin disks and light disks using bold red disks.
        \begin{figure}[bt]
            \centering
            \includegraphics[page=6]{gged_kclique_kconn.pdf}%
            \caption{Left: An instance of {\planarmonothreesat}, a Boolean
                formula
                $\phi = (x_1\lor x_2 \lor x_3) \land (\overline{x_2} \lor
                    \overline{x_3} \lor \overline{x_4}) \land (x_1\lor x_3 \lor x_4)
                    \land (\overline{x_1}\lor \overline{x_2} \lor \overline{x_4})$ with
                four variables and four clauses.  The satisfying assignment
                $(x_1,x_2,x_3,x_4) = (1,0,1,0)$ is represented by
                black lines. Right: Skeleton of the collection of
                disks $\D_\phi$ constructed from $G_\phi$. The arrows represent
                the movement of the light disks in the variable gadgets for
                the assignment $(1,0,1,0)$.}
            \label{fig:pm3sat_example}
        \end{figure}
        Our clause gadgets consist of many heavy disks that nearly connect the
        corresponding variable gadgets.  A truth assignment to a variable will
        correspond to a movement of the light disks inside the variable
        gadget such that the gadgets of those clauses that the assignment fulfils
        are connected to the variable gadget.
        Note that heavy disks connect consecutive variable gadgets (see again \Cref{fig:pm3sat_example}~(right)), so it suffices to move disks inside the variable gadgets in order to make the instance connected.

        Let $X_\phi=\{x_1,\ldots,x_n\}$, $m=|C_\phi|$, and let $\D_\phi$ be the set of unit disks that we will construct using~$G_\phi$.
        The set $\D_\phi$ consists of $n$ subsets $\X_i$, for $i\in [n]$ that represent the \emph{variable gadgets}.
        An example of a variable gadget~$\X_i$ (rotated clockwise by~$90^\circ$)
        is shown in \Cref{fig:variablegadget_example} (left).
        Assume that $x_i$ appears as a positive literal in $k_i$ clauses and as a negative literal in $k'_i$ clauses, and set $\kappa = \max(k_i,k'_i)$.
        The gadget $\X_i$ contains $\kappa+1$ light disks $x_{i, 0}, x_{i,1},\ldots,x_{i,\kappa}$ for positive clauses and $\kappa+1$ disks $\overline{x}_{i, 0},\overline{x}_{i,1},\ldots,\overline{x}_{i,\kappa}$ for negative clauses.
        Apart from the described light disks, a variable gadget has $2\kappa+2$ disconnected components and a \emph{skeleton} that consists of intersecting heavy disks.
        For $j\in [\kappa]$, there are two disconnected components containing a disk centred at $c(x_{i,j})\pm(4,0)$ (that is, one at $c(x_{i,j})+(4,0)$ and one at $c(x_{i,j})-(4,0)$).
        There are also two disconnected components containing a disk centred at $c(x_{i,0})\pm(0,4)$.
        Moreover, the skeleton is connected to disks centred at $c(x_{i,j})\pm(2,0)$ for each $x_{i,j}$ with $j \in [\kappa]$ and two disks centred at $c(x_{i,0})\pm(0,2)$.
        The same applies to the disks $\overline{x}_{i,j}$.
        Note that we can draw the variable gadget on a grid, thus all centres belong to $\mathbb{Z}^2$.

        Moving a light disk $x_{i,j}$ by $(\pm 3,0)$ for $j\in[\kappa]$, or $(0,\pm 3)$ for $j =0$, connects a disconnected component to the skeleton.
        The same holds for the disks $\overline{x}_{i,j}$.
        We move light disks to connect the $2\kappa+2$ components with the skeleton, and we define two movements that connect all components.
        We say that light disks are \emph{rotated counterclockwise} if $x_{i,j}$ is moved to $c(x_{i,j})+(-3,0)$ for $j\in  [\kappa]$, $\overline{x}_{i,0}$ to $c(\overline{x}_{i,0})+(0,-3)$, $\overline{x}_{i,j}$ to $c(\overline{x}_{i,j})+(3,0)$ for $j\in [\kappa]$ and lastly, $x_{i,0}$ to $c(x_{i,0})+(0,3)$.
        Analogously, the light disks are \emph{rotated clockwise} if we instead subtract the described movement vectors.
        Rotating disks clockwise and counterclockwise corresponds
        to assigning~$x_i$ to false and true, respectively.
        In particular, we observe the following.
        \begin{claim}
            \label{claim:twomov_movdisks}
            If the light disks in~$\X_i$ are all rotated clockwise or are
            all rotated counterclockwise, then the resulting intersection graph
            of $\X_i$ is connected.  Moreover, no other movement achieves
            connectivity with total moving distance of at most $3(2\kappa+2)$.
        \end{claim}
        \begin{proof}
            \renewcommand{\qedsymbol}{$\triangle$}
            Rotating (counter)clockwise makes each of the $2\kappa+2$ light disks connect the skeleton to a unique component of $\X_i$.
            Since $\X_i$ has $2\kappa+3$ components induced by heavy disks, the movement yields a graph in $\kconnected[1]$.

            Since each light disk moves by a distance of exactly $3$, the total moving distance is $3(2\kappa+2)$, and no other position within distance $3$ connects a disconnected component to the skeleton.
            After we move one light disk, the movements of all other light disks are forced to follow the same rotation to connect every component. Hence the only two possible movements are the clockwise and counterclockwise rotations.
            Therefore, no other movement achieves connectivity with total moving distance of at most $3(2\kappa+2)$.
        \end{proof}

        We set $m' = \sum_{i\in [n]} (2\max(k_i,k'_i)+2) \in O(m)$, the sum of the numbers of light disks of all variables.
        We place the variable gadgets horizontally according to the embedding of $G_\phi$, and connect the skeletons of consecutive variable gadgets with a chain of $3m'$ intersecting heavy disks.
        It follows that no light disk of a variable gadget can be moved to another variable gadget by a distance of at most $3m'$.
        We simulate clauses using claw-shaped connected components of heavy disks, placed following the structure of $G_\phi$.
        Since $G_\phi$ has a rectilinear embedding, we can define a component as a horizontal row of intersecting disks, with at most three vertical arms of connected disks going to the corresponding variable gadgets.
        For the variable gadget $\X_i$ and the $j$th positive clause containing $x_i$ as a literal, $i\in [n]$ and $j\in [k_i]$, the corresponding connected component has a disk in $c(x_{i,j})+(-3,1)$.
        Hence, the connected component connects to the variable gadget when light disks of $\X_i$ are rotated counterclockwise.
        Analogously, we place connected components for negative clauses from below.

        \begin{figure}[pbt]
            \centering
            \includegraphics[scale=1,page=7]{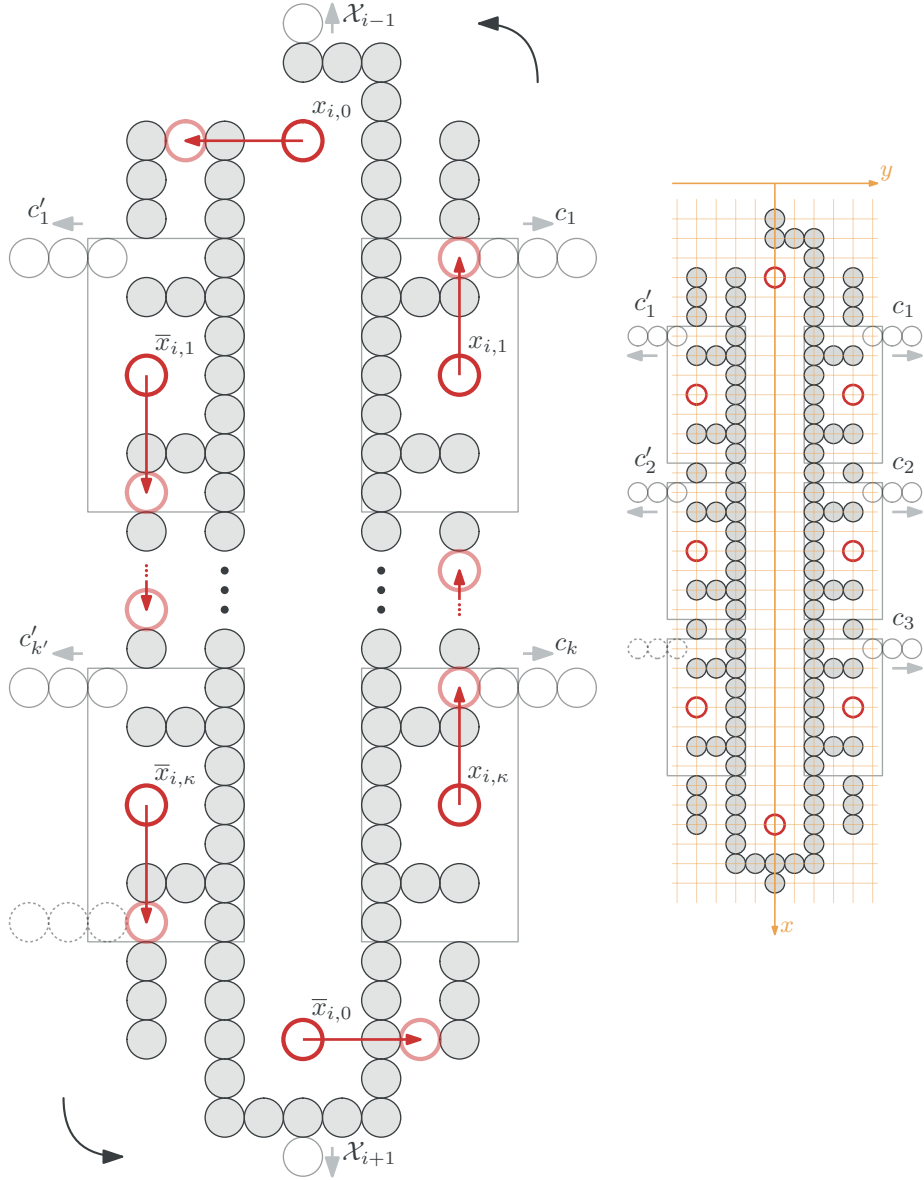}%
            \caption{Variable gadget $\X_i$ for $x_i$ (rotated by $90^\circ$). Left: There is one light disk corresponding to each clause that contains $x_i$ (note that $k'_i<k_i$ in this case). The arrows and light red bold disks describe the assignment of $x_i$ to true.
                The black shaded disks form the skeleton for $\X_i$ and the dashed grey disks indicate an unused slot for a negative clause.
                After the movement, all clauses that contain $x_i$ as a positive literal are connected to the variable gadget. Right: A variable gadget with $\kappa = \max(k_i,k'_i) = 3$. All disk centres lie on the integer grid (orange).
            }
            \label{fig:variablegadget_example}
        \end{figure}

        If $\phi$ is satisfiable, then each light disk is moved by a distance of $3$ according to the variable assignment.
        Thus, the total moving distance is $3m'$ for any assignment of the variables.
        Lastly, we set $W = 12m'$ and prove the following claim. The value $12m'$ is a sufficiently large weight to ensure that in a feasible solution for $\D_\phi$, no heavy disk is moved such that it reduces the number of connected components.

        \begin{claim}\label{claim:3sat_iff_ggedconn}
            There is an assignment to the variables $x_1,\dots,x_n$ that
            satisfies~$\phi$ if and only if there is a
            $\size{\D_\phi}$-tuple~$D$ of movement vectors such that
            $G_{\D_\phi + D} \in \kconnected[1]$ and $\|w \cdot D\|\le 3m'$.
        \end{claim}
        \begin{proof}
            \renewcommand{\qedsymbol}{$\triangle$}
            If there is an assignment to the variables $x_1,\ldots,x_n$ that satisfies $\phi$, then we rotate disks of $\X_i$ counterclockwise if $x_i$ is true and clockwise otherwise.
            By \Cref{claim:twomov_movdisks}, this connects all components of $\X_i$.
            Moreover, all components corresponding to clauses are connected to at least one variable gadget, given that all clauses are satisfied.
            Only light disks were moved, which yields a total moving distance of $3m'$.
            Hence, there is a $\size{\D_\phi}$-tuple~$D$ of movement vectors such that $G_{\D_\phi + D} \in \kconnected[1]$ and $\|w \cdot D\|\le 3m'$.

            Conversely, assume that such a $D$ exists.
            Notice that any heavy disk is at a distance of at least one unit from any other component.
            If a heavy disk is moved to connect a component, then its moving distance exceeds the threshold $3m'$ by the value $W$.
            Hence, we may assume that the moving distance for all heavy disks in $D$ is $0$.
            Since $G_{\D_\phi + D} \in \kconnected[1]$, the light disks must have been rotated (counter)clockwise by \Cref{claim:twomov_movdisks}, and all clause components must be connected to at least one variable gadget.
            By the same observation, the rotations are the only two movements that connect each variable gadget $\X_i$ with total moving distance of $3 \cdot (2\max(k_i,k'_i)+2)$, for each $i\in [n]$.
            Moreover, no light disk can be moved to a different variable gadget by a distance of at most $3m'$.
            Consequently, if $G_{\D_\phi + D} \in \kconnected[1]$ with $\|w \cdot D\|\le 3m'$, then $D$ describes exactly the rotations of light disks for all variable gadgets.
            We construct an assignment such that $x_i$ is true if disks in $\X_i$ were rotated counterclockwise and false otherwise.
            By the structure of $G_\phi$, the assignment satisfies $\phi$.
        \end{proof}
        The hardness follows from \Cref{claim:3sat_iff_ggedconn}.
        Following the structure of $G_\phi$, we place a polynomial number of disks in the variable gadgets and clause components to construct~$\D_\phi$.
        There are no cycles in variable gadgets and components for clauses, thus $G_{\D_\phi}$ is a forest.
        Moreover, the weights consist of two values: $1$ and $W$.
        Finally, all centres have integer coordinates and we describe the weights using polynomials in $n$ and $m$, which yields the strong \NP-hardness.
    \end{proof}
}

\end{document}